\documentclass[journal,onecolumn,12pt]{IEEEtran}
\usepackage{pslatex}
\usepackage{amsfonts,color,morefloats}
\usepackage{amssymb,amsmath,latexsym,amsthm}

\newcommand{\lcm}{{\rm lcm}}

\newcommand{\Z}{\mathbb{{Z}}}

\newcommand{\gf}{{\mathrm{GF}}}

\newcommand{\C}{{\mathcal{C}}}

\newcommand{\RM}{{\mathrm{RM}}}

\newcommand{\codenp}{\C_{(q,n,\delta,1)}}
\newcommand{\ocode}{\C_{(q,n,2\delta,n-\delta+1)}}
\newcommand{\gx}{g_{(q,n,2\delta,n-\delta+1)}(x)}
\newcommand{\gt}{\tilde{g}_{(q,n,2\delta,n-\delta+1)}(x)}
\newcommand{\tcode}{\tilde{\C}_{(q,n,2\delta,n-\delta+1)}}

\newcommand{\codeeven}{\C_{(q,n,2\delta,\frac{n}{2}-\de+1)}}

\newcommand{\codeodd}{\C_{(q,n,2\delta-1,\frac{n+1}{2}-\de+1)}}

\newcommand{\geven}{g_{\big(q, n, 2\delta, \frac n 2-\delta+1\big)}(x)}
\newcommand{\godd}{g_{\big(q, n, 2\delta-1, \frac{n+1}{2}-\delta+1\big)}(x)}

\newcommand{\Fq}{\gf(q)}
\newcommand{\Fqm}{\gf(q^m)}
\newcommand{\Ftm}{\gf(2^m)}
\newcommand{\Ft}{\gf(2)}

\newcommand{\al}{\alpha}

\newcommand{\de}{\delta}
\newcommand{\lam}{\lambda}
\newcommand{\ep}{\epsilon}
\newcommand{\supp}{\text{supp}}
\newcommand{\ol}{\overline}
\newcommand{\lf}{\lfloor}
\newcommand{\rf}{\rfloor}
\newcommand{\pr}{\prime}

\newtheorem{theorem}{Theorem}
\newtheorem{lemma}[theorem]{Lemma}
\newtheorem{proposition}[theorem]{Proposition}
\newtheorem{corollary}[theorem]{Corollary}
\newtheorem{conj}{Conjecture}

\newtheorem{example}{Example}

\newtheorem{remark}{Remark}
\newtheorem{result}{Result}

\begin{document}

\title{Two Families of LCD BCH codes}

\author{Shuxing~Li, Chengju~Li,~Cunsheng~Ding,~Hao~Liu~

\thanks{C. Ding's research was supported by the Hong Kong Research Grants Council, under Grant No. 16301114.}

\thanks{S. Li is with the Department of Mathematics, Hong Kong University of Science and Technology,
                                                  Clear Water Bay, Kowloon, Hong Kong, China  (email: lsxlsxlsx1987@gmail.com).}

\thanks{C. Li is with the School of Computer Science and Software Engineering, East China Normal University,
Shanghai, 200062, China (email: lichengju1987@163.com).}

\thanks{C. Ding is with the Department of Computer Science and Engineering, The Hong Kong University of Science and Technology,
 Clear Water Bay, Kowloon, Hong Kong, China (email: cding@ust.hk).}

\thanks{H. Liu is with the Department of Computer Science
and Engineering, The Hong Kong University of Science and Technology,
 Clear Water Bay, Kowloon, Hong Kong, China (email: hliuar@connect.ust.hk).}
}
\date{\today}
\maketitle

\begin{abstract}

Historically, LCD cyclic codes were referred to as reversible cyclic codes, which had application in data storage. Due to a newly discovered application in cryptography, there has been renewed interest on LCD codes. In this paper, we explore two special families of LCD cyclic codes, which are both BCH codes. The dimensions and the minimum distances of these LCD BCH codes are investigated.  As a byproduct, the parameters of some primitive BCH codes are also obtained.

\end{abstract}


\section{Introduction}
Let $\gf(q)$ be a finite field of size $q$. An $[n,k, d]$ linear code $\mathcal C$ over $\gf(q)$ is a linear subspace of $\gf(q)^n$ with dimension $k$ and
minimum distance $d$. A linear code $\mathcal C$ over $\gf(q)$ is called an \emph{LCD
code (linear code with complementary dual)} \cite{Massey1} if
$\mathcal C \cap \mathcal C^\perp=\{\textbf{0}\}$, where $\mathcal C^\perp$ denotes the dual code of $\mathcal C$ and is defined by
$$\mathcal C^\perp=\{(b_0, b_1, \ldots, b_{n-1}) \in \gf(q)^n: \sum_{i=0}^{n-1}b_ic_i=0 \mbox{ for all } (c_0, c_1, \ldots, c_{n-1}) \in \mathcal C\}.$$

An $[n,k,d]$ linear code $\mathcal C$ is called \emph{cyclic} if
$(c_0, c_1, \ldots, c_{n-1}) \in \mathcal C$ implies
$(c_{n-1}, c_0, c_1, \ldots, c_{n-2}) \in \mathcal C$.
By identifying each vector $(c_0, c_1, \ldots, c_{n-1}) \in \gf(q)^n$ with
$$c_0+c_1x+c_2x^2+\cdots+c_{n-1}x^{n-1} \in \gf(q)[x]/(x^n-1),$$
a linear code $\mathcal C$ of length $n$ over $\gf(q)$ corresponds to a $\gf(q)$-submodule of
$\gf(q)[x]/(x^n-1)$. $\mathcal C$ is a cyclic code if and only if the
corresponding submodule is an ideal of $\gf(q)[x]/(x^n-1)$. Note that every
ideal of $\gf(q)[x]/(x^n-1)$ is principal. Then there is a monic polynomial $g(x)$ of the smallest
degree such that $\mathcal C=\langle g(x) \rangle$ and $g(x) \mid (x^n-1)$. In addition, $g(x)$ is unique and called the \emph{generator
polynomial}, and $h(x)=(x^n-1)/g(x)$ is referred to as the \emph{parity-check polynomial} of $\mathcal C$.

Let $f(x) \in \Fq[x]$ be a monic polynomial with degree $l$, then the reciprocal polynomial of $f$ is defined to be $x^lf(x^{-1})$. $f$ is called self-reciprocal if $f(x)$ is equal to its reciprocal. A cyclic code $\C$ with generator polynomial $g(x)$ is called \emph{reversible} if $g(x)$ is self-reciprocal. The reversibility implies if $(c_0,c_1,\ldots,c_{n-1}) \in \C$, then $(c_{n-1},c_{n-2},\ldots,c_0) \in \C$.
We have the following lemma, showing that LCD cyclic codes and reversible cyclic codes are the same thing.

\begin{lemma}{\rm(\cite{YM94}, see also \cite[Theorem 4]{LDL})} \label{TC} Let $\mathcal C$ be a cyclic code over $\gf(q)$ with generator polynomial $g(x)$. Then the following statements are
equivalent.
\begin{enumerate}
  \item $\mathcal C$ is an LCD code.
  \item $g(x)$ is self-reciprocal, i.e., $\C$ is a reversible cyclic codes.
  \item $\beta^{-1}$ is a root of $g(x)$ for every root $\beta$ of $g(x)$.
\end{enumerate}
\end{lemma}

LCD cyclic codes were first studied by Massey for the data storage applications \cite{Ma}, under the name of \emph{reversible codes}. Massey showed that some LCD cyclic codes are BCH codes, and
made a comparison between LCD codes and non-LCD codes \cite{Ma}. He also demonstrated that
asymptotically good LCD codes exist \cite{Massey1}. Yang and Massey gave a necessary and
sufficient condition for a cyclic code to have a complementary dual \cite{YM94}. Using the hull
dimension spectra of linear codes, Sendrier showed that LCD codes meet the asymptotic Gilbert-Varshamov
bound \cite{Sendr}. Esmaeili and Yari analysed LCD codes that are quasi-cyclic \cite{EY09}.
Muttoo and Lal constructed an LCD cyclic code over $\gf(q)$ \cite{ML86}.  Tzeng and Hartmann proved
that the minimum distance of a class of LCD cyclic codes is greater than the BCH bound \cite{TH}.  Dougherty, Kim, Ozkaya, Sok and Sol\'{e} developed a linear programming bound on the largest size of
an LCD code of given length and minimum distance \cite{DKOSS}. Carlet and Guilley investigated an
application of LCD codes against side-channel attacks, and presented several constructions of LCD
codes \cite{CG}. There are two well known classes of LCD cyclic codes \cite[p. 206]{MS}, which are Melas's double-error correcting binary codes with parameters $[2^m-1, 2^m-2m-1, d \ge 5]$ and Zetterberg's double-error correcting binary codes of length $2^\ell+1$. A well-rounded treatment of reversible cyclic codes was given in \cite{LDL}. In addition, Boonniyoma and Jitman gave a study on linear codes with Hermitian complementary dual \cite{BJ}.

The objective of this paper is to investigate the basic parameters of two families of LCD primitive BCH codes, including their dimensions and minimum distances. As a byproduct, the parameters of several classes of primitive BCH codes are also obtained. According to the tables of
best known linear codes (referred to as the \emph{Database} later) maintained by Markus Grassl at http://www.codetables.de/ and the tables of best cyclic codes documented in \cite{Dingbk15}, some of the codes presented in this paper are optimal in the sense that they have the best possible parameters.

\section{$q$-cyclotomic cosets and BCH codes}

In this section, we introduce $q$-cyclotomic cosets and their coset leaders, which will play a crucial role in our analysis of LCD codes. Moreover, we give a brief review on BCH codes.

\subsection{$q$-cyclotomic cosets}

To deal with cyclic codes of length $n$ over $\gf(q)$, we need to study the canonical factorization of $x^n-1$
over $\gf(q)$. To this end, we are going to introduce $q$-cyclotomic cosets modulo $n$. Note that $x^n-1$ has no
repeated factors over $\gf(q)$ if and only if $\gcd(n, q)=1$. Throughout this paper, we assume
that $\gcd(n, q)=1$.

Let $\Bbb Z_n = \{0,1,2, \cdots, n-1\}$ denote the ring of integers modulo $n$. For each $s \in \Z_n$, the \emph{$q$-cyclotomic coset of $s$ modulo $n$\index{$q$-cyclotomic coset modulo $n$}} is defined by
$$
C_s=\{s, sq, sq^2, \cdots, sq^{\ell_s-1}\} \bmod n \subseteq \Z_n,
$$
where $\ell_s$ is the smallest positive integer such that $q^{\ell_s}s \equiv s \pmod{n}$. Therefore, $\ell_s$ is the size of the
$q$-cyclotomic coset $C_s$. We use $cl(s)$ to denote the coset leader of $C_{s}$, which is the smallest integer belonging to $C_{s}$. Note that the subscript of $C_s$ is regarded as an integer modulo $n$. Thus, we have $C_{-s}=C_{n-s}$.

\subsection{BCH codes}

Let $n$ be a positive integer with $\gcd(n,q)=1$ and $m$ be the smallest positive integer such that $q^m \equiv 1 \bmod n$. Let $\alpha$ be a generator of $\gf(q^m)^*$ and put $\beta=\alpha^{\frac {q^m-1} n}$. Then $\beta$ is a primitive $n$-th root of unity. For $0 \le i \le n-1$, let $m_i(x)$ denote the minimal polynomial of $\beta^i$ over $\gf(q)$. We use $i \bmod n$ to denote the unique integer in the set $\{0, 1, \ldots, n-1\}$, which is congruent to $i$ modulo $n$. Thus, we have $m_i(x):=m_{i \bmod n}(x)$.

For an integer $\delta \ge 2$, define
\begin{equation*}
g_{(q,n,\delta,b)}(x)=\lcm(m_{b}(x), m_{b+1}(x), \cdots, m_{b+\delta-2}(x)),
\end{equation*}
where $\lcm$ denotes the least common multiple of these polynomials. Let $\mathcal C_{(q,n,\delta,b)}$ denote the cyclic code of length $n$ with generator polynomial $g_{(q,n,\delta,b)}(x)$. Then $\mathcal C_{(q,n,\delta,b)}$ is called a BCH code with \emph{designed distance} $\delta$. The \emph{BCH bound} implies that the minimum distance of $\mathcal C_{(q,n,\delta,b)}$ is greater than or equal to the designed distance $\de$. We call $\mathcal C_{(q,n,\delta,b)}$  a \emph{narrow-sense BCH code} if $b=1$. When $n=q^m-1$, $\mathcal C_{(q,n,\delta,b)}$ is called a \emph{primitive BCH code}.

So far, we have very limited knowledge of BCH codes, as the dimension and minimum distance of BCH codes are in general open. The narrow-sense primitive BCH codes form the most well-studied subclass of BCH codes, which have been investigated in a series of literature, including \cite{AKS,ACS92,AS94,Ber15,Berlekamp,Charp90,Charp98,YH96,D,DDZ15,KL72,Mann,MS,YF}. The reader is referred to \cite{DDZ15} for a recent survey on known results of narrow-sense primitive BCH codes and to \cite{LDXG} for some new results on narrow-sense nonprimitive BCH codes. As pointed out by Charpin in \cite{Charp98}, it is very difficult to determine the minimum distance of BCH codes. However, in some special cases, the minimum distance is known.

\begin{lemma} {\rm \cite[p. 247]{BBFKKW}} \label{Lemmamd}
For a narrow-sense BCH code $\mathcal C_{(q,n,\delta,1)}$  over $\text{GF}(q)$ of length $n$ with designed distance $\delta$, its
minimum distance $d=\delta$ if $\delta$ divides $n$.
\end{lemma}

The following corollary is a generalization of Lemma \ref{Lemmamd} and will be employed later.

\begin{corollary} \label{corollarygener} Let $\mathcal C_{(q,n,\delta,b)}$ be the BCH code over $\text{GF}(q)$ of length
$n$ with designed distance $\delta$. Then its
minimum distance $d=\delta$ if $\delta$ divides $\gcd(n, b-1)$. \end{corollary}

\begin{proof}
Denote $$c(x)=\frac {x^n-1} {x^{n/\delta}-1}=x^{(\delta-1)\frac n \delta}+\cdots+x^{\frac n \delta}+1.$$
Since $\delta \mid (b-1)$, we have $c(\beta^j)=0$ for each $b \le j \le b+\delta-2$ and $\delta \nmid j$,
where $\beta$ is a primitive $n$-th root of unity. It then follows that $c(x) \in \mathcal C_{(q,n,\delta,b)}$. It is clear
that the Hamming weight of $c(x)$ is equal to $\delta$.
\end{proof}

\section{Two families of LCD primitive BCH codes}

In this section, we introduce two families of LCD primitive BCH codes, whose parameters will be analyzed subsequently. From now on, we always assume that $n=q^m-1$. We always use $\bar n$ to denote $\lceil \frac{n}{2} \rceil$ and $\bar m$ to denote $\lceil \frac{m}{2} \rceil$.

For each integer $\delta$ with $2 \le \delta \le \lf \frac{n+1}{2} \rf$, define
\begin{equation} \label{LCDGPoly} g(x)=\begin{cases}
\lcm \Big(x+1, g_{(q, n, \delta, \frac n 2+1)}(x), g_{\big(q, n, \delta, \frac n 2-(\delta-1)\big)}(x)\Big), & \text{ if  $n$ is even;}\\
\lcm \Big(g_{(q, n, \delta, \frac {n+1} 2)}(x), g_{\big(q, n, \delta, \frac {n+1} 2-(\delta-1)\big)}(x)\Big), & \text{ if  $n$ is odd.}\end{cases}\end{equation}
It can be verified that \begin{equation} \label{GeneratorPoly} g(x)=\begin{cases}
g_{\big(q, n, 2\delta, \frac n 2-(\delta-1)\big)}(x), & \text{ if  $n$ is even;}\\
g_{\big(q, n, 2\delta-1, \frac {n+1} 2-(\delta-1)\big)}(x), & \text{ if  $n$ is odd.}\end{cases}\end{equation}
Let $\codeeven$ (resp. $\codeodd$) be the BCH code of length $n$ with the generator polynomial $g_{\big(q, n, 2\delta, \frac n 2-(\delta-1)\big)}(x)$ (resp. $g_{\big(q, n, 2\delta-1, \frac {n+1} 2-(\delta-1)\big)}(x)$). Note that $2 \leq \delta \leq \lf \frac{n+1}{2} \rf$ ensures $g(x) \ne x^n-1$. Thus, $\codeeven \ne \{\textbf{0}\}$ and $\codeodd \ne \{\textbf{0}\}$. It is easy to check that $\geven$  and $\godd$ are self-reciprocal. Therefore, it follows from Lemma \ref{TC} that $\codeeven$ and $\codeodd$ are LCD BCH codes.

For each $2 \leq \delta < \lf \frac{n+1}{2} \rf$, define
$$
\tilde{g}_{(q,n,2\delta,n-\delta+1)}(x)=\lcm(g_{(q,n,\delta,1)}(x),g_{(q,n,\delta,n-\de+1)}(x)).
$$
Let $\tcode$ denote the cyclic code of length $n$ with generator polynomial $\gt$. By Lemma \ref{TC}$, \tcode$ is an LCD cyclic code. For the minimum distance $d$ of $\tcode$, it was shown in \cite{TH} that
$$
\begin{cases}
d=\de & \mbox{if $\de \mid n$},\\
d \ge \de+1 & \mbox{otherwise}.
\end{cases}
$$
Moreover, if we consider the even-like subcode of $\tcode$, namely, the code $\ocode$ with length $n$ and generator polynomial
$$
\gx=(x-1)\gt,
$$
its minimum distance is at least $2\de$ by the BCH bound. Hence, a potentially great improvement on the minimum distance is expected by considering the even-like subcode of $\tcode$. This intuition motivates us to study the code $\ocode$, which is an LCD BCH code.

We remark that the above two families of codes are closely related. In fact, when $q$ is odd, $\codeeven$ and $\ocode$ are monomially equivalent \cite[p. 24]{HP}. Let $\al$ be the primitive element of $\gf(q^m)$. Note that $\codeeven$ has generator polynomial $\geven$. The parity-check matrix of $\codeeven$ consists of rows with the form
\begin{align*}
&(1,\al^{\frac{n}{2}+j},\al^{2(\frac{n}{2}+j)},\al^{3(\frac{n}{2}+j)},\ldots,\al^{(n-2)(\frac{n}{2}+j)},\al^{(n-1)(\frac{n}{2}+j)}) \\
=&(1,-\al^j,\al^{2j},-\al^{3j},\ldots,\al^{(n-2)j},-\al^{(n-1)j})
\end{align*}
where $-\de+1 \le j \le \de-1$. Meanwhile, the code $C_{(q,n,2\de,n-\de+1)}$ has generator polynomial $g_{(q,n,2\delta, n-\delta+1)}(x)$. The parity check matrix of $C_{(q,n,2\de,n-\de+1)}$ consists of rows with the form
$$
(1,\al^j,\al^{2j},\al^{3j},\ldots,\al^{(n-2)j},\al^{(n-1)j})
$$
where $-\de+1 \le j \le \de-1$. Hence, the parity-check matrix of $C_{(q,n,2\de,n-\de+1)}$ can be obtained from that of $\codeeven$, by multiplying $-1$ in some columns. Thus, $\codeeven$ and $C_{(q,n,2\de,n-\de+1)}$ are monomially equivalent when $q$ is odd. Consequently, they have the same parameters, including the dimension and minimum distance. It is worthy to note that this equivalence is generally not true when $q$ is even.

\section{Parameters of the primitive narrow-sense BCH codes $\mathcal C_{(q,n,\delta,1)}$}

In this section, we always assume that $u$ is an integer with $1 \le u \le q-1$.

\begin{lemma} \label{AKSYF} {\rm(\cite[Lemmas 8 and 9]{AKS}, \cite[Theorem 3]{YH96})}
Let $m \ge 2$. Then we have the following.
\begin{itemize}
\item[1)] When $m$ is odd, for $1 \le j \le q^{(m+1)/2}$, $|C_j|=|C_{-j}|=m$. For $1 \le j \le q^{(m+1)/2}$, $j$ is a coset leader of a $q$-cyclotomic coset if and only if $q \nmid j$.

\item[2)] When $m$ is even, $|C_{q^{m/2}+1}|=|C_{-q^{m/2}-1}|=\frac{m}{2}$ and $|C_j|=|C_{-j}|=m$ for $1 \le j \le 2q^{m/2}$, $j \ne q^{m/2}+1$. For $1 \le j \le 2q^{m/2}$, $j$ is a coset leader of a $q$-cyclotomic coset if and only if $q \nmid j$.
\end{itemize}
\end{lemma}

We present the size of each cyclotomic coset $C_j$ and characterize all coset leaders $j$ satisfying $1 \le j  \le uq^{\bar m}$ in the following proposition, where $m \ge 5$ is an odd integer.

\begin{proposition} \label{PO}
 Let $m \ge 5$ be an odd integer and let $j$ be an integer with $1 \le j \le u q^{\bar m}$ and $q \nmid j$, where $1 \le u \le q-1$.
 Then the following holds.
 \begin{enumerate}
   \item $|C_j|=m$,
   \item $j$ is a coset leader of the cyclotomic coset $C_j$ except $j \in J_1 \cup J_2$, where
\begin{equation} \label{J1} J_1=\{j_{\bar m}q^{\bar m}+j_1q+j_0: 1 \le j_{\bar m} \le u-1, 0 \le j_1 < j_{\bar m}, 1 \le j_0 \le q-1\}\end{equation} and
\begin{equation} \label{J2}J_2=\{j_{\bar m}q^{\bar m}+j_{\bar m-1}q^{\bar m-1}+j_0: 1 \le j_{\bar m} \le u-1, 1 \le j_{\bar m-1} \le q-1, 1 \le j_0 \le j_{\bar m}\}.\end{equation}
   \item $|J_1 \cup J_2|=(u^2-u)(q-1)$.
 \end{enumerate}
\end{proposition}

\begin{proof} For each $j$ with $1 \le j \le u q^{\bar m}$, let $\ell=|C_j|$. Since $m$ is odd, we have $1 \le \ell \le \frac m 3$ if $\ell <m$.
For $m \ge 9$, one can check that
$$j < jq^\ell < n \text{ for all $1 \le j \le u q^{\bar m}$},$$
which means that
$$jq^\ell \equiv j \pmod n$$ does not hold for each $\ell <m$. Thus we have $|C_j|=m$ if $m \ge 9$.

For $m\in \{5,7\}$, if $|C_j| < m$, then $|C_j|=1$. Therefore, $qj \equiv j \bmod {n}$, which means that $j \equiv 0 \bmod{\frac{q^m-1}{q-1}}$. This is impossible as $j<\frac{q^m-1}{q-1}$. Hence, $|C_j|=m$. 

Below we characterize all coset leaders $j$ satisfying $1 \le j  \le uq^{\bar m}$. To this end, we have to find all integers $j$ satisfying
$j \in C_i$, i.e.,
\begin{equation}\label{IJinSC}jq^\ell \bmod n=i\end{equation} for some integer $\ell$ with $1 \le \ell \le m-1$ and some integer $i<j$.
 Let $i$ and $j$ be two integers with $q \nmid i, q \nmid j$, and $i<j \le uq^{\bar m}$.
By Lemma \ref{AKSYF}, $j$ is a coset leader if $1 \le j \le q^{\bar m}$ and $q \nmid j$, so we can further assume that $j \ge q^{\bar m}+1$.
Then we have the two $q$-adic expansions
$$i=i_{\bar m}q^{\bar m}+i_{\bar m-1}q^{\bar m-1}+\cdots+i_1q+i_0$$ and
$$j=j_{\bar m}q^{\bar m}+j_{\bar m-1}q^{\bar m-1}+\cdots+j_1q+j_0,$$
where $1 \le i_0, j_0 \le q-1$, $1\le j_{\bar m} \le u-1$, and $0 \le i_{\bar m} \le j_{\bar m}$.

 \emph{Case 1:} When $1 \le \ell \le \bar m-2$, it is easy to check that $i < jq^\ell <n$, so \eqref{IJinSC} does not hold.

\emph{Case 2:} When $\ell=\bar m-1$, we have
  $$jq^\ell=j_{\bar m}q^m+j_{\bar m-1}q^{m-1}+\cdots+j_1q^{\bar m}+j_0q^{\bar m-1}$$
  by noting that $\bar m=\frac {m+1} 2$. Then
  $$jq^\ell \bmod n=j_{\bar m-1}q^{m-1}+\cdots+j_2q^{\bar m+1}+j_1q^{\bar m}+j_0q^{\bar m-1}+j_{\bar m}.$$
By \eqref{IJinSC}, we obtain
\begin{equation} \label{JO1} j_{\bar m}=i_0, \  j_{\bar m-1}=j_{\bar m-2}=\cdots=j_2=i_{\bar m-2}=i_{\bar m -3}=\cdots=i_1=0, \ j_1=i_{\bar m}, \ j_0=i_{\bar m-1}. \end{equation}
 Thus $j=j_{\bar m}q^{\bar m}+j_1q+j_0$.

 Notice that $i < j$. Then $i_{\bar m} \le j_{\bar m}$. We assert that the equality $i_{\bar m} = j_{\bar m}$ does not hold. Otherwise, it follows from \eqref{JO1}
 and $i<j$ that $i_{\bar m-1} \le j_{\bar m-1}=0$ and $j_0=i_{\bar m-1}=0$, which is a contradiction. We then deduce that
$0 \le j_1=i_{\bar m}< j_{\bar m} \le u-1$. Denote
$$J_1=\{j_{\bar m}q^{\bar m}+j_1q+j_0: 1 \le j_{\bar m} \le u-1, 0 \le j_1 < j_{\bar m}, 1 \le j_0 \le q-1\}.$$
Then when $\ell=\bar m-1$, \eqref{IJinSC} holds if and only if $j \in J_1$ for $\ell=\bar m-1$.

\emph{Case 3:} When $\ell=\bar m$, we have
$$jq^\ell=j_{\bar m}q^{m+1}+j_{\bar m-1}q^m+\cdots+j_1q^{\bar m+1}+j_0q^{\bar m}.$$
Then $$jq^\ell \bmod n =j_{\bar m-2}q^{m-1}+\cdots+j_1q^{\bar m+1}+j_0q^{\bar m}+j_{\bar m}q+j_{\bar m-1}.$$
By \eqref{IJinSC}, we obtain
\begin{equation} \label{JO2} j_{\bar m}=i_1, \ j_{\bar m-1}=i_0, \ j_{\bar m-2}=\cdots=j_2=j_1=i_{\bar m-1}=i_{\bar m-2}=\cdots=i_{2}=0, \ j_0=i_{\bar m}.\end{equation}
 Thus $j=j_{\bar m}q^{\bar m}+j_{\bar m-1}q^{\bar m-1}+j_0$.

\emph{ Case 3.1:} If $i_{\bar m}< j_{\bar m}$, it then follows from \eqref{JO2} that \eqref{IJinSC} holds if and only if $j \in J_{21}$, where
$$J_{21}=\{j_{\bar m}q^{\bar m}+j_{\bar m-1}q^{\bar m-1}+j_0: 1 \le j_{\bar m} \le u-1, 1 \le j_{\bar m-1} \le q-1, 1 \le j_0<j_{\bar m}\}.$$

\emph{ Case 3.2:} If $i_{\bar m}=j_{\bar m}$ and $j_{\bar m-1}>0$, it then follows from \eqref{JO2} that \eqref{IJinSC} holds if and only if $j \in J_{22}$, where
$$J_{22}=\{j_{\bar m}q^{\bar m}+j_{\bar m-1}q^{\bar m-1}+j_0: 1 \le j_{\bar m} \le u-1, 1 \le j_{\bar m-1} \le q-1, j_0=j_{\bar m}\ge1\}.$$

\emph{ Case 3.3:} If $i_{\bar m}=j_{\bar m}$, $j_{\bar m-1}=0$, then $i_0=j_{\bar m-1}=0$. This is a contradiction to the assumption that $1 \le i_0 \le q-1$.

Denote
$$J_2=J_{21} \cup J_{22}=\{j_{\bar m}q^{\bar m}+j_{\bar m-1}q^{\bar m-1}+j_0: 1 \le j_{\bar m} \le u-1, 1 \le j_{\bar m-1} \le q-1, 1 \le j_0 \le j_{\bar m}\}.$$
Then when $\ell=\bar m$, \eqref{IJinSC} holds if and only if $j \in J_2$ for $\ell=\bar m$.

\emph{Case 4:} When $\bar m+1 \le \ell \le m-1$, let $\ell=\bar m+\epsilon$, where $1 \le \epsilon \le \bar m-2$. Then we have
$$jq^\ell=j_{\bar m}q^{2\bar m+\epsilon}+\cdots+j_{\bar m-\epsilon-1}q^m+j_{\bar m-\epsilon-2}q^{m-1}+\cdots+j_0q^{\bar m+\epsilon}.$$
Then $$jq^\ell \bmod n =j_{\bar m-\epsilon-2}q^{m-1}+\cdots+j_0q^{\bar m+\epsilon}+j_{\bar m}q^{\epsilon+1}+\cdots+j_{\bar m-\epsilon-1}.$$
Note that $j_0 \ge 1$. Then $jq^\ell \bmod n > i$, which implies that \eqref{IJinSC} is impossible in this case.

Combining Cases 1, 2, 3, and 4, we obtain the conclusion on the characterization of coset leaders.
Note that $|J_1|=|J_2|=\frac {u(u-1)} 2(q-1)$. Since $J_1 \cap J_2=\emptyset$, we $|J_1 \cup J_2|=(u^2-u)(q-1)$.
\end{proof}

Employing Proposition~\ref{PO}, we obtain the dimension of certain narrow-sense primitive BCH code.

\begin{theorem} Let $m \ge 5$ be an odd integer and $\delta=uq^{\frac {m+1} 2}+1$, where $1 \le u \le q-1$.
Then the code $\mathcal C_{(q,n, \delta, 1)}$ has length $n$, dimension
$$k=q^m-1-(uq^{\frac {m-1} 2}-u^2+u)(q-1)m,$$ and minimum distance $d \ge \delta$.
Furthermore, the generator polynomial is given by $$g_{(q,n, \delta, 1)}(x)=\prod_{\substack {1 \le j \le uq^{\frac {m+1} 2} \\ q \nmid j, j \not \in J_1 \cup J_2}}m_j(x),$$
where $J_1$ and $J_2$ are defined in Proposition~\ref{PO}.
\end{theorem}

\begin{proof}
The desired conclusions follow from Proposition \ref{PO} and the BCH bound immediately.
\end{proof}

\begin{example}
\begin{enumerate}
  \item When $(q, m, u)=(2, 5, 1)$ in the above theorem, the code $\mathcal C_{(q,n, \delta, 1)}$ has parameters $[31, 11, 11]$, which is an optimal code according to the Database.
  \item When $(q, m, u)=(2, 7, 1)$ in the above theorem, the code $\mathcal C_{(q,n, \delta, 1)}$ has parameters $[127, 71, 19]$, which are the best parameters for linear codes according to the Database.
\end{enumerate}
\end{example}

The following proposition gives the size of each cyclotomic coset $C_j$ and  characterizes all coset leaders $j$ satisfying $1 \le j  \le uq^{\bar m}$, where $m \ge 2$ is an even integer.

\begin{proposition} \label{PE}
Let $m \ge 2$ be an even integer and let $j$ be an integer with $1 \le j \le u q^{\bar m}$ and $q \nmid j$, where $1 \le u \le q-1$. Then the following holds.
 \begin{enumerate}
   \item $|C_j|=m$, except $|C_{v(q^{\bar m}+1)}|=\bar m$, where $1 \le v \le u-1$.
   \item $j$ is a coset leader of the cyclotomic coset $C_j$ except $j \in J$, where
\begin{equation} \label{JE} J=\{j_{\bar m}q^{\bar m}+j_0: 1\le j_0< j_{\bar m} \le u-1\}.\end{equation}
   \item $|J|=\frac {(u-1)(u-2)} 2$.
 \end{enumerate}
\end{proposition}

\begin{proof}
 Let $i$ and $j$ be two integers with $q \nmid i, q \nmid j$, and $i<j \le uq^{\bar m}$. Suppose that $j \in C_i$. Then there exists some integer $\ell$ with $1 \le \ell \le m-1$ such that
 \begin{equation}\label{IJinSC2}jq^\ell \bmod n=i.\end{equation}
By Lemma \ref{AKSYF}, $j$ is a coset leader if $1 \le j \le 2q^{\bar m}$ and $q \nmid j$, so we can further assume that $j \ge 2q^{\bar m}+1$.
Then we have the two $q$-adic expansions
$$i=i_{\bar m}q^{\bar m}+i_{\bar m-1}q^{\bar m-1}+\cdots+i_1q+i_0$$ and
$$j=j_{\bar m}q^{\bar m}+j_{\bar m-1}q^{\bar m-1}+\cdots+j_1q+j_0,$$
where $1 \le i_0, j_0 \le q-1$, $2 \le j_{\bar m} \le u-1$, and $0 \le i_{\bar m} \le j_{\bar m}$.

\emph{Case 1:} When $1 \le \ell \le \bar m-1$, it is easy to check that $i < jq^\ell <n$, so \eqref{IJinSC2} does not hold.

\emph{Case 2:} When $\ell=\bar m$, we have
  $$jq^\ell=j_{\bar m}q^m+j_{\bar m-1}q^{m-1}+\cdots+j_1q^{\bar m+1}+j_0q^{\bar m}.$$
 Then
  $$jq^\ell \bmod n=j_{\bar m-1}q^{m-1}+\cdots+j_1q^{\bar m+1}+j_0q^{\bar m}+j_{\bar m}.$$
By \eqref{IJinSC2}, we obtain
\begin{equation} \label{JE1} j_{\bar m}=i_0, \ j_{\bar m-1}=j_{\bar m-2}=\cdots=j_1=i_{\bar m-1}=i_{\bar m-2}\cdots=i_1=0, \ j_0=i_{\bar m}. \end{equation}
 Thus $j=j_{\bar m}q^{\bar m}+j_0$.

 \emph{ Case 2.1:} If $i_{\bar m}< j_{\bar m}$, it then follows from \eqref{JE1} that \eqref{IJinSC2} holds if and only if $j \in J$, where
$$J=\{j_{\bar m}q^{\bar m}+j_0: 1 \le j_0 < j_{\bar m} \le u-1\}.$$

\emph{ Case 2.2:} If $i_{\bar m}=j_{\bar m}$, since $i_{\bar m-1}=\cdots=i_1=0$, we have $i_0<j_0$. Then
$$i_0<j_0=i_{\bar m}=j_{\bar m}=i_0,$$
which is a contradiction. Thus \eqref{IJinSC2} does not hold.

\emph{Case 3:} When $\bar m+1 \le \ell \le m-1$, let $\ell=\bar m+\epsilon$, where $1 \le \epsilon \le \bar m-1$. Then
$$jq^\ell=j_{\bar m}q^{2\bar m+\epsilon}+\cdots+j_{\bar m-\epsilon}q^m+j_{\bar m-\epsilon-1}q^{m-1}+\cdots+j_0q^{\bar m+\epsilon}.$$
Then $$jq^\ell \bmod n =j_{\bar m-\epsilon-1}q^{m-1}+\cdots+j_1q^{\bar m+\epsilon+1}+j_0q^{\bar m+\epsilon}+j_{\bar m}q^{\epsilon}+j_{\bar m-1}q^{\epsilon-1}+\cdots+j_{\bar m-\epsilon}.$$
Note that $j_0 \ge 1$. Then $jq^\ell \bmod n > i$, which implies that \eqref{IJinSC2} is impossible in this case.

Summarizing all the discussions in Cases 1, 2, and 3, we get the desired conclusion of 2). It is easy to see that
$$jq^\ell \bmod n > j$$ in both Cases 1 and 3. Then we have
$|C_j|=\bar m$ if $|C_j|< m$. Moreover, it follows from Case 2 that
$$j=j_{\bar m}q^{\bar m}+j_0=j_0q^{\bar m}+j_{\bar m} \text{ and } j_0=j_{\bar m}.$$
Then we proved 1). It is clear that $|J|=\frac {(u-1)(u-2)} 2$. This completes the proof.
\end{proof}

Employing Proposition~\ref{PE}, we can obtain the dimension of certain narrow-sense primitive BCH code.

\begin{theorem} Let $m \ge 2$ be an even integer and $\delta=uq^{\frac m 2}+1$. Then the code $\mathcal C_{(q,n, \delta, 1)}$ has length
$n$, dimension $$k=q^m-1-uq^{\frac m 2-1}(q-1)m+\frac {(u-1)^2} 2 m,$$ and minimum distance $d \ge \delta$. When $u=1$, we have $d =\delta$.
Furthermore, the generator polynomial is given by $$g_{(q,n, \delta, 1)}(x)=\prod_{\substack {1 \le j \le uq^{\frac m 2} \\ q \nmid j, j \not \in J}}m_j(x),$$
where $J$ is defined in Proposition~\ref{PE}.
\end{theorem}

\begin{proof}  When $u=1$, it is clear that $\delta | n$.
The desired conclusions then follow from Lemma \ref{Lemmamd}, Proposition \ref{PE} and the BCH bound.
\end{proof}

\begin{example}
\begin{enumerate}
  \item When $(q, m, u)=(2, 4, 1)$ in the above theorem, the code $\mathcal C_{(q,n, \delta, 1)}$ has parameters $[15, 7, 5]$, which is an optimal code according to the Database.
  \item When $(q, m, u)=(3, 4, 1)$ or $(q, m, u)=(3, 4, 2)$ in the above theorem, the code $\mathcal C_{(q,n, \delta, 1)}$ has parameters $[80, 56, 10]$ and $[80, 34, 20]$, respectively.
 The former has the best parameters for linear codes according to the Database.
\end{enumerate}
\end{example}

\section{Parameters of LCD BCH code $\codeeven$ when $q$ is odd}

In this section, we always assume that $q$ is odd, $u$ is an integer with $1 \le u \le q-1$. The following proposition will be used later.
\begin{proposition} \label{PRelation} Let $q$ be odd and $m \ge 2$. Then we have the following.
\begin{enumerate}
  \item $|C_{\bar n+i}|=|C_i|=|C_{-i}|=|C_{\bar n-i}|$.
  \item $|C_{\bar n+qi}|=|C_{\bar n+i}|$ and $|C_{\bar n-qi}|=|C_{\bar n-i}|$.
  \item $C_i=C_j$ if and only if $C_{\bar n+i}=C_{\bar n+j}$.
  \item $C_i=C_j$ if and only if $C_{\bar n-i}=C_{\bar n-j}$.
\end{enumerate}
\end{proposition}

\begin{proof}
The proof of 1) and 2) is obvious. Note that $q$ is odd. It is clear that
$$\frac n 2 \pm i \equiv (\frac n 2 \pm j)q^\ell \pmod n,$$
 is equivalent to
$$i \equiv jq^\ell \pmod n$$ for each $\ell$ with $0 \le \ell \le m-1$.
Then the conclusions of 3) and 4) follow.
\end{proof}

Let $1 \le u \le q-1$ be an integer. Define
$$J_{(q,n, u)}^+=\bigcup_{1 \le j \le uq^{\bar m}}C_{\bar n+j} \mbox{ and } J_{(q,n, u)}^-=\bigcup_{1 \le j \le uq^{\bar m}}C_{\bar n-j},$$
where $q$ is odd.
It can be deduced from Proposition \ref{PRelation} that $C_{\bar n+i} \ne C_{\bar n+j}$ and $C_{\bar n-i} \ne C_{\bar n-j}$ if and only if
$C_i \ne C_j$. The following corollary then follows from Propositions \ref{PO} and \ref{PE} directly.

\begin{corollary} \label{CorollaryQO} Let $q$ be odd and $j$ be an integer with $1 \le j \le uq^{\bar m}$.
\begin{enumerate}
  \item If $m \ge 5$ is odd, then $|C_{\bar n+j}|=|C_{\bar n-j}|=m$ and
  $$|J_{(q,n, u)}^+|=|J_{(q,n, u)}^-|=(uq^{\bar m-1}-u^2+u)(q-1)m.$$
  \item If $m \ge 2$ is even, then $|C_{\bar n+j}|=|C_{\bar n-j}|=m$ except $j=v(q^{\bar m}+1)$ with
$|C_{\bar n+v(q^{\bar m}+1)}|=|C_{\bar n-v(q^{\bar m}+1)}|=\frac m 2$, where $v=1, 2, \ldots, u-1$. In this case,
$$|J_{(q,n, u)}^+|=|J_{(q,n, u)}^-|=uq^{\bar m-1}(q-1)m-\frac {(u-1)^2} 2 m.$$
\end{enumerate}
\end{corollary}

\begin{theorem} \label{THEOREMQO}
Let $m \ge 2$ be an integer and $\delta=uq^{\bar m }+1$.
\begin{enumerate}
  \item If $m \ge 5$ is odd, then $\mathcal C_{(q,n, \delta, \frac n 2+1)}$ and $\mathcal C_{(q,n, \delta, \frac n 2-(\delta-1))}$ both have length
$n$, dimension $$k=q^m-1-(uq^{\frac {m-1} 2}-u^2+u)(q-1)m,$$ and minimum distance $d \ge \delta$.
In addition, the generator polynomials are given by $$g_{(q,n, \delta, \frac n 2+1)}(x)=\prod_{\substack {1 \le j \le uq^{\frac {m+1} 2} \\ q \nmid j, j \not \in J_1 \cup J_2}}m_{\frac n 2+j}(x)
\text{ and } g_{(q,n, \delta, \frac n 2-(\delta-1))}(x)=\prod_{\substack {1 \le j \le uq^{\frac {m+1} 2} \\ q \nmid j, j \not \in J_1 \cup J_2}}m_{\frac n 2-j}(x),$$
where $J_1$ and $J_2$ are defined in Proposition~\ref{PO}.
  \item If $m \ge 2$ is even, then $\mathcal C_{(q,n, \delta, \frac n 2+1)}$ and $\mathcal C_{(q,n, \delta, \frac n 2-(\delta-1))}$ both have length
$n$, dimension $$q^m-1-uq^{\frac m 2-1}(q-1)m+\frac {(u-1)^2} 2 m,$$ and minimum distance $d \ge \delta$.
In addition, the generator polynomials are given by $$g_{(q,n, \delta, \frac n 2+1)}(x)=\prod_{\substack {1 \le j \le uq^{\frac m 2} \\ q \nmid j, j \not \in J}}m_{\frac n 2+j}(x)
\text{ and } g_{(q,n, \delta, \frac n 2-(\delta-1))}(x)=\prod_{\substack {1 \le j \le uq^{\frac m 2} \\ q \nmid j, j \not \in J}}m_{\frac n 2-j}(x),$$
where $J$ is defined in Proposition~\ref{PE}.
\end{enumerate}
\end{theorem}

\begin{proof} The proof follows from Corollary \ref{CorollaryQO} and the BCH bound, and is omitted here.
\end{proof}

\begin{example}
\begin{enumerate}
  \item When $(q, m, u)=(3, 5, 1), (3, 5, 2)$ in the above theorem, the code $\mathcal C_{(q,n, \delta, 1)}$ has parameters $[242, 152, d \ge 28]$ and $[242, 82, d \ge 55]$, respectively.
  \item When $(q, m, u)=(4, 4, 1),(4, 4, 2),(4, 4, 3)$ in the above theorem, the code $\mathcal C_{(q,n, \delta, 1)}$ has parameters $[255, 207, d \ge 17]$, $[242, 161, d \ge 33]$, and $[242, 119, d \ge 49]$, respectively.
\end{enumerate}
\end{example}

\emph{A. Parameters of $\codeeven$ when $m$ is odd}

The following proposition plays an important role in determining the dimension of the BCH code $\codeeven$ when $m \ge 5$ is odd and $\delta=uq^{\frac {m+1} 2}+1$, where $1 \le u \le q-1$.

\begin{proposition} \label{PI1}
For odd $m \ge 5$, we have
$$J_{(q,n,u)}^+ \cap J_{(q,n,u)}^-=\bigcup_{l \in \mathcal J_O}(C_{\bar n+l} \cup C_{\bar n-l}),$$
where the union is disjoint and
$$\mathcal J_O=\{l_{\bar m}q^{\bar m}+l_{\bar m-1}q^{\bar m-1}+(q-1)\sum_{i=1}^{\bar m-2}q^i+l_0:
0 \le l_{\bar m} \le u-1, 0 \le l_{\bar m-1} \le q-2, q-u \le l_0 \le q-1\}.$$
Moreover, $$|J_{(q,n,u)}^+ \cap J_{(q,n,u)}^-|=2u^2(q-1)m.$$
\end{proposition}

\begin{proof}
We are going to find the integers $i$ and $j$ with $1 \le i \le uq^{\bar m}$ and $1 \le j \le uq^{\bar m}$ such that
$$
C_{\bar n +i}=C_{\bar n -j}.
$$
This is equivalent to
\begin{equation}\label{EIJ1} \bar n +i \equiv (\bar n-j)q^\ell \pmod n \text{  and } i+jq^\ell \equiv 0 \pmod n\end{equation} for
some $1 \le \ell \le m-1$.

By Proposition \ref{PRelation}, we can further assume that $q \nmid i$ and $q \nmid j$. Then we have the $q$-adic expansions
   $$i=i_{\bar m}q^{\bar m}+i_{\bar m-1}q^{\bar m-1}+\cdots+i_1q+i_0$$ and
$$j=j_{\bar m}q^{\bar m}+j_{\bar m-1}q^{\bar m-1}+\cdots+j_1q+j_0,$$
where $0 \le i_{\bar m}, j_{\bar m} \le u-1$, $1 \le i_0, j_0 \le q-1$, and $0 \le i_k, j_k \le q-1$ for all $k$ with $1 \le k \le \bar m-1$.

  \emph{Case 1:} When $1 \le \ell \le \bar m-2$, it is easy to check that $0 < i+jq^\ell < n$ by noticing
  that $j_{\bar m} \le u-1 <q-1$, so $i+jq^\ell \equiv 0 \pmod n$ does not hold.

   \emph{Case 2:} When $\ell=\bar m-1$, it can be verified that $i+jq^\ell \equiv \Delta \pmod n$, where
$$\Delta=j_{\bar m-1}q^{m-1}+\cdots+j_2q^{\bar m+1}+(j_1+i_{\bar m})q^{\bar m}+(j_0+i_{\bar m-1})q^{\bar m-1}+
i_{\bar m-2}q^{\bar m-2}+\cdots+i_1q+(i_0+j_{\bar m}).$$
It is clear that $0< \Delta < 2n$. It then follows from \eqref{EIJ1} that $\Delta=n$. Thus
$$j_{\bar m-1}=\cdots=j_2=j_1+i_{\bar m}=j_0+i_{\bar m-1}=i_{\bar m-2}=\cdots=i_1=i_0+j_{\bar m}=q-1.$$
Then $$i=i_{\bar m}q^{\bar m}+i_{\bar m-1}q^{\bar m-1}+(q-1)(q^{\bar m-2}+\cdots+q^2+q)+i_0,$$
where $$0 \le i_{\bar m} \le u-1, \ 0 \le i_{\bar m-1} \le q-2, \text{ and } q-u \le i_0 \le q-1.$$
Hence, there exists exactly one integer $j$ with $1 \le j \le uq^{\bar {m}}$, such that
$$
C_{\bar n+i}=C_{\bar n-j},
$$
if and only if $i$ has the above form.
Therefore,
$$J_{(q,n,u)}^+ \cap J_{(q,n,u)}^- \supset \bigcup_{l \in \mathcal J_O} C_{\bar n+l}.$$

   \emph{Case 3:} When $\ell=\bar m$, we have $i+jq^\ell \equiv \Delta \pmod n$, where
$$\Delta=j_{\bar m-2}q^{m-1}+\cdots+j_1q^{\bar m+1}+(j_0+i_{\bar m})q^{\bar m}+i_{\bar m-1}q^{\bar m-1}+
\cdots+i_2q^2+(j_{\bar m}+i_1)q+(j_{\bar m-1}+i_0).$$
Notice that $0< \Delta < 2n$. It then follows from \eqref{EIJ1} that $\Delta=n$. Thus
$$j_{\bar m-2}=\cdots=j_1=j_0+i_{\bar m}=i_{\bar m-1}=\cdots=i_2=j_{\bar m}+i_1=j_{\bar m-1}+i_0=q-1.$$
Then $$j=j_{\bar m}q^{\bar m}+j_{\bar m-1}q^{\bar m-1}+(q-1)(q^{\bar m-2}+\cdots+q^2+q)+j_0,$$
where $$0 \le j_{\bar m} \le u-1, \ 0 \le j_{\bar m-1} \le q-2, \text{ and } q-u \le j_0 \le q-1.$$
Hence, there exists exactly one integer $i$ with $1 \le i \le uq^{\bar {m}}$, such that
$$
C_{\bar n +i}=C_{\bar n -j},
$$
if and only if $j$ has the above form. Therefore,
$$J_{(q,n,u)}^+ \cap J_{(q,n,u)}^- \supset \bigcup_{l \in \mathcal J_O} C_{\bar n-l}.$$

   \emph{Case 4:} When $\bar m+1 \le \ell \le m-1$, denote $\ell=\bar m+\epsilon$, where $1 \le \epsilon \le \bar m-2$. Then
  $i+jq^\ell \equiv \Delta \pmod n$, where
  \begin{eqnarray*}\Delta &=&j_{\bar m-\epsilon-2}q^{m-1}+\cdots+j_0q^{\bar m+\epsilon}+i_{\bar m}q^{\bar m}+i_{\bar m-1}q^{\bar m-1}+\cdots+i_{\epsilon+2}q^{\epsilon+2}\\ &+&
 (i_{\epsilon+1}+j_{\bar m})q^{\epsilon+1}+\cdots+(i_1+j_{\bar m-\epsilon})q+(i_0+j_{\bar m-\epsilon-1}). \end{eqnarray*}
   It is easy to see that the coefficient of $q^{\bar m}$ in the $q$-adic expansion of $\Delta$ is less than $q-1$.
 Thus we have $0<\Delta <n$, which means that \eqref{EIJ1} is impossible.

Note that Cases 1, 2, 3, and 4 contain all possible pairs $(i,j)$, such that $1 \le i,j \le uq^{\bar m}$ and $C_{\bar n+i}=C_{\bar n-j}$. Thus, we have $J_{(q,n,u)}^+ \cap J_{(q,n,u)}^-=\bigcup_{l \in \mathcal J_O}(C_{\bar n+l} \cup C_{\bar n-l})$. Next, we are going to show that this union is disjoint. By Proposition~\ref{PO}, each $l \in \mathcal J_O$ is a coset leader and $|C_{\bar n+l}|=|C_{\bar n-l}|=m$. Hence, by Proposition~\ref{PRelation}, we have $C_{\bar n+l} \ne C_{\bar n+l^{\prime}}$ and $C_{\bar n-l} \ne C_{\bar n-l^{\prime}}$ for distinct $l,l^{\prime} \in \mathcal J_O$. In addition, suppose $C_{\bar n+l} = C_{\bar n-l^{\prime}}$. If $l \in \mathcal J_O$, by the arguments in Case 2, we have
$$
l^{\prime}=l_{\bar m}^{\prime}q^{\bar m}+(q-1)(q^{\bar m-1}+\cdots+q^2)+l_1^{\prime}q+l_0^{\prime},
$$
where
$$
0 \le l_{\bar m}^{\prime} \le u-1, q-u \le l_1^{\prime} \le q-1, 1 \le l_0^{\prime} \le q-1.
$$
Hence, $l^{\prime} \not\in \mathcal J_O$. Similarly, if $l^{\prime} \in \mathcal J_O$, by the arguments in Case 3, we must have $l \not\in \mathcal J_O$. Therefore, the union $\bigcup_{l \in \mathcal J_O}(C_{\bar n+l} \cup C_{\bar n-l})$ is disjoint and $|J_{(q,n,u)}^+ \cap J_{(q,n,u)}^-|=2m|\mathcal J_O|=2u^2(q-1)m$.
\end{proof}

\begin{remark}\label{rem-1}
Let $m \ge 5$ be an odd integer. Let $1 \le i,j \le uq^{\bar m}$ be two integers with $q$-adic expansions
$$
i=i_{\bar m}q^{\bar m}+i_{\bar m-1}q^{\bar m-1}+\cdots+i_1q+i_0
$$
and
$$
j=j_{\bar m}q^{\bar m}+j_{\bar m-1}q^{\bar m-1}+\cdots+j_1q+j_0.
$$
The proof of Proposition~\ref{PI1} shows that there exists an unique $1 \le j \le uq^{\bar m}$, such that
$$
i+jq^\ell \equiv 0 \pmod n
$$
for some $1 \le \ell \le m-1$, if and only if one of the following holds:
\begin{itemize}
\item $i \in \mathcal J_O, j \not\in J_1 \cup J_2 \cup \mathcal J_O$
with
$$
j_{\bar m-1}=\cdots=j_2=j_1+i_{\bar m}=j_0+i_{\bar m-1}=i_{\bar m-2}=\cdots=i_1=i_0+j_{\bar m}=q-1.
$$

\item $i \not\in J_1 \cup J_2 \cup \mathcal J_O, j \in \mathcal J_O$
with
$$
j_{\bar m-2}=\cdots=j_1=j_0+i_{\bar m}=i_{\bar m-1}=\cdots=i_2=j_{\bar m}+i_1=j_{\bar m-1}+i_0=q-1.
$$
\end{itemize}
We remark that this result does not depend on the parity of $q$ and $n$. Namely, the above result is true when $q$ is odd, $n$ is even or $q$ is even, $n$ is odd.
\end{remark}

The following result gives the dimension of the LCD code $\codeeven$ when $m \ge 5$ is odd and $\delta=uq^{\frac {m+1} 2}+1$, where $1 \le u \le q-1$.

\begin{theorem} \label{theoremLCDQOMO} Let $m \ge 5$ be an odd integer, $q$ odd,  and $\delta=uq^{\frac {m+1} 2}+1$, where $1 \le u \le q-1$.
Then $\codeeven$  has length $n$, dimension
$$k=q^m-2-2(uq^{\frac {m-1} 2}-2u^2+u)(q-1)m,$$ and minimum distance $d \ge 2\delta$.
In addition, the generator polynomial is given by \begin{equation} \label{GQOMO} g(x)=(x+1)\prod_{\substack {1 \le l \le uq^{\frac {m+1} 2} \\ q \nmid l, l \not \in J_1 \cup J_2\cup \mathcal J_O} }m_{\frac n 2+l}(x)m_{\frac n 2-l}(x),\end{equation}
where $J_1$, $J_2$ are defined in Proposition~\ref{PO} and $\mathcal J_O$ is defined in Proposition~\ref{PI1}.
\end{theorem}

\begin{proof}
Let $1 \le i,j \le uq^{\frac{m+1}{2}}$ be two integers satisfying $C_{\frac{n}{2}+i}=C_{\frac{n}{2}-j}$. By Remark~\ref{rem-1}, we must have either $i \in \mathcal J_O$, $j \not\in J_1 \cup J_2 \cup \mathcal J_O$ or $j \in \mathcal J_O$, $i \not\in J_1 \cup J_2 \cup \mathcal J_O$. Together with Theorem \ref{THEOREMQO}, we can see that the generator polynomial is given by \eqref{GQOMO}, and its degree is equal to $1+2(uq^{\frac {m-1} 2}-2u^2+u)(q-1)m$. Hence, the dimension follows easily. The minimum distance $d \ge 2\delta$ follows from the BCH bound.
\end{proof}

\begin{example}
When $(q, m, u)=(3, 7, 1),(3, 7, 2)$ in the above theorem, the code $\codeeven$ has parameters $[2186, 1457, d \ge 164]$, and $[2186, 841, d \ge 326]$, respectively.
\end{example}

\emph{B. Parameters of $\codeeven$ when $m$ is even}

To investigate the parameters of the LCD BCH code $\codeeven$ when $m \ge 2$ is even, we will need
the following conclusion.

\begin{proposition} \label{PI2}
Let $m \ge 2$ be an even number. Suppose
$$
\begin{cases}
  1 \le u \le \frac{q-1}{2} & \mbox{if $m=2$,} \\
  1 \le u \le q-1 & \mbox{if $m \ge 4$.}
\end{cases}
$$
Then we have
$$J_{(q,n,u)}^+ \cap J_{(q,n,u)}^-=\bigcup_{l \in {\mathcal J}_E} C_{\bar n-l},$$
where the union is disjoint and $${\mathcal J}_E=\{l_{\bar m}q^{\bar m}+(q-1)(q^{\bar m-1}+q^{\bar m-2}+\cdots+q)+l_0: 0 \le l_{\bar m} \le u-1 \text{ and } q-u \le l_0 \le q-1\}.$$
Moreover,
$$|J_{(q,n,u)}^+ \cap J_{(q,n,u)}^-|=u^2m.$$
\end{proposition}

\begin{proof}
We are going to find all the integers $i$ and $j$ with $1 \le i \le uq^{\bar m}$ and $1 \le j \le uq^{\bar m}$ such that
$$
C_{\bar n +i}=C_{\bar n -j}.
$$
This is equivalent to
\begin{equation*} \bar n +i \equiv (\bar n-j)q^\ell \pmod n \text{  and } i+jq^\ell \equiv 0 \pmod n\end{equation*} for
some $1 \le \ell \le m-1$.

By Proposition \ref{PRelation}, we can assume that $q \nmid i$ and $q \nmid j$.
For $i, j \le uq^{\bar m}$, let $$i=i_{\bar m}q^{\bar m}+i_{\bar m-1}q^{\bar m-1}+\cdots+i_1q+i_0$$ and
$$j=j_{\bar m}q^{\bar m}+j_{\bar m-1}q^{\bar m-1}+\cdots+j_1q+j_0,$$
where $0 \le i_{\bar m}, j_{\bar m}\le u-1$, $1 \le i_0, j_0 \le q-1$, and $0 \le i_k, j_k \le q-1$ for all $1 \le k \le \bar m-1$.

\emph{Case 1:} When $1 \le \ell \le \bar m-1$, we can easily see that $0 < i+jq^\ell < n$ as $j_{\bar m} \le u-1 < q-1$, which implies
  that $i+jq^\ell \equiv 0 \pmod n$ does not hold.

\emph{Case 2:} When $\ell=\bar m$, it can be verified that
  \begin{eqnarray*}
  i+jq^\ell \equiv \Delta \pmod n,
  \end{eqnarray*}
  where $$\Delta=j_{\bar m-1}q^{m-1}+\cdots+j_1q^{\bar m+1}+(j_0+i_{\bar m})q^{\bar m}+i_{\bar m-1}q^{\bar m-1}+\cdots+i_1q+(i_0+j_{\bar m}).$$
  Notice that $0 <\Delta < 2n$. If $\Delta \equiv 0 \pmod n$, then $\Delta=n$ and
$$j_{\bar m-1}=\cdots=j_1=j_0+i_{\bar m}=i_{\bar m-1}=\cdots=i_1=i_0+j_{\bar m}=q-1.$$
Thus $$j=j_{\bar m}q^{\bar m}+(q-1)(q^{\bar m-1}+q^{\bar m-2}+\cdots+q)+j_0,$$ where
$$0 \le j_{\bar m} \le u-1 \text{ and } q-u \le j_0 \le q-1.$$
Hence, there exists exactly one integer $i$ with $1 \le i \le uq^{\bar {m}}$, such that
$$
C_{\bar n +i}=C_{\bar n -j},
$$
if and only if $j$ has the above form. Therefore,
$$J_{(q,n,u)}^+ \cap J_{(q,n,u)}^- \supset \bigcup_{l \in \mathcal J_E} C_{\bar n-l}.$$

\emph{Case 3:} When $\bar m+1 \le \ell \le m-1$, let $\ell=\bar m+\epsilon$, where $1 \le \epsilon \le \bar m-1$.
Then one can check that $i+jq^\ell \equiv \Delta \pmod n$, where
$$\Delta=j_{\bar m-\epsilon-1}q^{m-1}+\cdots+j_0q^{\bar m+\epsilon}+i_{\bar m}q^{\bar m}+\cdots+
i_{\epsilon+1}q^{\epsilon+1}+(i_\epsilon+j_{\bar m})q^{\epsilon}+\cdots+(i_0+j_{\bar m-\epsilon}).$$
Note that the coefficient of $q^{\bar m}$ in the $q$-adic expansion of $\Delta$ is equal to $i_{\bar m} \le u-1 < q-1$. Then $0 < \Delta < n$, which means that
$$(i+jq^\ell) \bmod n=\Delta \not \equiv 0 \pmod n.$$

Note that Cases 1, 2 and 3 deal with all possible pairs $(i,j)$, such that $1 \le i,j \le uq^{\bar m}$ and $C_{\bar n+i}=C_{\bar n-j}$. Thus, $J_{(q,n,u)}^+ \cap J_{(q,n,u)}^-=\bigcup_{l \in \mathcal J_E} C_{\bar n-l}$. By Proposition~\ref{PE}, each $l \in \mathcal J_E$ is a coset leader and $|C_{\bar n-l}|=m$. In particular, when $m=2$, we need $1 \le u \le \frac{q-1}{2}$ to ensure that each $l \in \mathcal J_E$ is a coset leader and $|C_{\bar n-l}|=m$. Hence, by Proposition~\ref{PRelation}, we have $C_{\bar n-l} \ne C_{\bar n-l^{\prime}}$ for distinct $l,l^{\prime} \in \mathcal J_E$.  Therefore, the union $\bigcup_{l \in {\mathcal J}_E} C_{\bar n-l}$ is disjoint and $|J_{(q,n,u)}^+ \cap J_{(q,n,u)}^-|=m|\mathcal J_E|=u^2m$.
\end{proof}

\begin{remark}\label{rem-2}
Let $m \ge 2$ be an even number. Suppose
$$
\begin{cases}
  1 \le u \le \frac{q-1}{2} & \mbox{if $m=2$,} \\
  1 \le u \le q-1 & \mbox{if $m \ge 4$.}
\end{cases}
$$
Let $1 \le i,j \le uq^{\bar m}$ be two integers with $q$-adic expansions
$$
i=i_{\bar m}q^{\bar m}+i_{\bar m-1}q^{\bar m-1}+\cdots+i_1q+i_0
$$
and
$$
j=j_{\bar m}q^{\bar m}+j_{\bar m-1}q^{\bar m-1}+\cdots+j_1q+j_0.
$$
The proof of Proposition~\ref{PI2} shows that for $1 \le j \le uq^{\bar m}$, there exists a unique $1 \le i \le uq^{\bar m}$, such that
$$
i+jq^\ell \equiv 0 \pmod n
$$
for some $1 \le \ell \le m-1$, if and only if $i, j \in \mathcal J_E$ with
$$
j_{\bar m-1}=\cdots=j_1=j_0+i_{\bar m}=i_{\bar m-1}=\cdots=i_1=i_0+j_{\bar m}=q-1.
$$
We remark that this result does not depend on the parity of $q$ and $n$. Namely, the above result is true when $q$ is odd, $n$ is even or $q$ is even, $n$ is odd.
\end{remark}

\begin{theorem} \label{TheoremCQOME} Let $q$ be odd and $m \ge 2$ be even. Let $\delta=uq^{\frac m 2}+1$, where
$$
\begin{cases}
  1 \le u \le \frac{q-1}{2} & \mbox{if $m=2$,} \\
  1 \le u \le q-1 & \mbox{if $m \ge 4$.}
\end{cases}
$$
Then $\codeeven$ has length $n$, dimension $$k=q^m-2-2uq^{\frac m 2-1}(q-1)m+(2u^2-2u+1)m,$$ and minimum distance $d \ge 2\delta$.
In addition, the generator polynomial is given by $$g(x)=(x+1)\prod_{\substack {1 \le l \le uq^{\frac m 2} \\ q \nmid l, l \not \in J}}m_{\frac n 2+l}(x)\prod_{\substack {1 \le l \le uq^{\frac m 2} \\ q \nmid l, l \not \in J \cup \mathcal J_E}}m_{\frac n 2-l}(x),$$
where $J$ is defined in Proposition~\ref{PE} and $\mathcal J_E$ is defined in Proposition~\ref{PI2}, respectively.
\end{theorem}
\begin{proof}
By Remark~\ref{rem-2}, if for $1 \le i,j \le u\delta^{\frac{m}{2}}$, $C_{\bar n+i}=C_{\bar n-j}$, then $i,j \in \mathcal J_E$. Note that $J \cap \mathcal J_E = \emptyset$. The dimension and the generator polynomial follow from Theorem \ref{THEOREMQO} and Proposition \ref{PI2}. The minimum distance $d \ge 2\delta$ follows from the BCH bound.
\end{proof}

\begin{example}
When $(q, m, u)=(5, 2, 1)$ in the above theorem, the code $\C_{5,24,12,7}$ has parameters $[24, 9, 12]$, which are the best parameters for linear codes according to the Database.
\end{example}

\begin{corollary}
Let $u=1$ and $\delta=q^{\frac m 2}+1$, where $q \equiv 3 \pmod 4$ and $m \equiv 2 \pmod 4$. Then the true minimum distance of the code $\codeeven$ presented in Theorem \ref{TheoremCQOME} is equal to $2\delta$.
\end{corollary}

\begin{proof}
Note that $b=\frac n 2-\delta+1$. It is easy to check that $2 \delta \mid \gcd(n, b-1)$ in this case. The desired result then
follows from Corollary \ref{corollarygener}.
\end{proof}

\begin{example}
When $(q, m, u)=(7, 2, 1)$ in the above corollary, the code $\C_{(7,48,16,17)}$ has parameters $[48, 25, 16]$, which are the best parameters for linear codes according to the Database.
\end{example}

\emph{C. Parameters of $\codeeven$ with designed distance $q^t-1$, where $1 \le t \le \bar m$}

The dimension of the LCD code $\codeeven$ is described in the following theorem when $\codeeven$ has designed distance $2\delta=q^t-1$ for an integer $t$ with $1 \le t \le \bar m$.

\begin{theorem} \label{theoremLCDDDQTMO} Let $q$ be odd and $m \ge 2$.
Suppose $\codeeven$ has designed distance $2\delta=q^t-1$, where $1 \le t \le \bar m$. Then $\codeeven$ has length
$n$, dimension
$$k=q^m-2-(q^t-q^{t-1}-2)m$$
and minimum distance $d \ge q^t-1$.
\end{theorem}

\begin{proof}
Set $\delta=\frac {q^t-1} 2$. Recall that the generator polynomial of the code $\codeeven$ is $g_{(q, n, 2\delta, \frac n 2-(\delta-1))}(x)$, we have
$$
\deg(g_{(q, n, 2\delta, \frac n 2-(\delta-1))}(x))=1+|\Big(\bigcup_{1 \le j \le \delta-1}C_{\bar n+j}\Big) \bigcap \Big(\bigcup_{1 \le j \le \delta-1}C_{\bar n-j}\Big)|
$$

It follows from Propositions \ref{PI1} and \ref{PI2} that
\begin{equation}\label{Eempty}\Big(\bigcup_{1 \le j \le \delta-1}C_{\bar n+j}\Big) \bigcap \Big(\bigcup_{1 \le j \le \delta-1}C_{\bar n-j}\Big)=\emptyset
\end{equation}
for each integer $m$ with $m \ge 2$ and $m \ne 3$. Using Remark~\ref{rem-1}, it can be checked that \eqref{Eempty} also holds for $m=3$. It then follows from Lemma~\ref{AKSYF} that 
$$
\deg(g_{(q, n, 2\delta, \frac n 2-(\delta-1))}(x))=(q^t-q^{t-1}-2)m+1 \text{ for } 1 \le t \le \bar m.
$$
Thus, the dimension is obtained. Moreover, $d \ge q^t-1$ follows from the BCH bound.
\end{proof}

\begin{example}
\begin{enumerate}
  \item When $(q, m, t)=(3, 5, 1), (3, 5, 2), (3, 5, 3)$ in the above theorem, the code $\codeeven$ has parameters $[242, 241, 2]$, $[242, 221, 8]$ and $[242, 161, 26]$, respectively. All of them are the best known parameters for linear codes according to the Database.
  \item When $(q, m, t)=(3, 4, 1), (3, 4, 3)$ in the above theorem, the code $\codeeven$ has parameters $[80, 79, 2]$ and $[80, 63, 8]$, respectively. Both of them are the best known parameters for linear codes according to the Database.
\end{enumerate}
\end{example}

In the above theorem, each triple $(q,m,t)$ satisfying $(q,m) \in \{(3,2),(3,3),(3,4),(3,5),(5,2),(7,2)\}$ and $1 \le t \le \bar m$ has been tested in numerical experiments and the experimental results suggest the following conjecture.

\begin{conj}
The code $\codeeven$ in Theorem \ref{theoremLCDDDQTMO} has true minimum distance $q^t-1$.
\end{conj}

\section{Parameters of LCD BCH code $\codeodd$ when $q$ is even}

In this section, we always assume that $q$ is even and $u$ is an integer with $1 \le u \le q-1$. The following proposition will be used later.

\begin{proposition} \label{PRQE} Let $q$ be even and $m \ge 2$. Then we have the following.
\begin{enumerate}
  \item $|C_{\bar n+i}|=|C_{2i+1}|$ and $|C_{\bar n-i}|=|C_{2i-1}|$.
  \item $C_{2i+1}=C_{2j+1}$ if and only if $C_{\bar n+i}=C_{\bar n+j}$.
  \item $C_{2i-1}=C_{2j-1}$ if and only if $C_{\bar n-i}=C_{\bar n-j}$.
\end{enumerate}
\end{proposition}
\begin{proof}
The proof of 1) is trivial. Since $q$ is even and $\gcd(2, n)=1$, it is clear that
$$\frac {n+1} 2 \pm i \equiv (\frac {n+1} 2 \pm j)q^\ell \pmod n,$$
which is equivalent to
$$2i \pm 1 \equiv (2j \pm 1)q^\ell \pmod n$$ for each $\ell$ with $0 \le \ell \le m-1$.
Conclusions 2) and 3) then follow.
\end{proof}

Let $1 \le u \le q-1$ be an integer. Define
$$\tilde J_{(q,n, u)}^+=\bigcup_{0 \le j \le uq^{\bar m}/2-1}C_{\bar n+j} \mbox{ and } \tilde J_{(q,n, u)}^-=\bigcup_{1 \le j \le uq^{\bar m}/2}C_{\bar n-j},$$
where $q$ is even.

\emph{A. Parameters of $\codeodd$ when $m$ is odd}

In this subsection, we always assume that $m \ge 5$ and $m$ is odd. It can be deduced from Proposition \ref{PRQE} that $C_{\bar n+i} \ne C_{\bar n+j}$ if and only if
$C_{2i+1} \ne C_{2j+1}$ (resp. $C_{\bar n-i} \ne C_{\bar n-j}$ if and only if
$C_{2i-1} \ne C_{2j-1}$). Let $J_1$ and $J_2$ be the sets of integers that are not coset leaders, which are given by \eqref{J1} and \eqref{J2}.
Note that $1 \le 2j+1 \le uq^{\bar m}-1$ if $0 \le j \le uq^{\bar m}/2-1$ and $1 \le 2j-1 \le uq^{\bar m}-1$ if $1 \le j \le uq^{\bar m}/2$.
Therefore, we have
$$
|\tilde J_{(q,n, u)}^+|=|\tilde J_{(q,n, u)}^-|=|\bigcup_{\substack{ 1 \le l \le uq^{\bar m}-1 \\ l \; odd}}C_{l}|.
$$
By Proposition~\ref{PO}, we have $|C_l|=m$ for each $1 \le l \le uq^{\bar m}-1$. When $m \ge 5$, by the definition of $J_1$ and $J_2$ in Proposition~\ref{PO}, if $j \in J_1 \cup J_2$, then $|C_j \cap (J_1 \cup J_2)|=1$. Thus, we have
\begin{align*}
|\tilde J_{(q,n, u)}^+|=|\tilde J_{(q,n, u)}^-|=&m|\{ 1 \le l \le uq^{\bar m}-1 \mid \mbox{$l$ is an odd coset leader}\}|\\
&+m|\{ 1 \le l \le uq^{\bar m}-1 \mid \mbox{$l \in J_1 \cup J_2$ is odd and $cl(l)$ is even}\}|.
\end{align*}
Define
\begin{equation} \label{Lambda1} \lambda_1:=\lambda_1(u,\bar m)=|\{1 \le  \tilde{j} \le uq^{\bar m}-1:  \tilde{j} \text{ is an odd coset leader }\}| \end{equation}
and
\begin{equation} \label{Lambda2} \lambda_2:=\lambda_2(u,\bar m)=|\{1 \le \tilde{j} \le uq^{\bar m}-1:  \tilde{j} \in J_1 \cup J_2  \text{ is odd, $cl(l)$ is even}\}|. \end{equation}
It then follows that
\begin{equation}\label{LL} |\tilde J_{(q,n, u)}^+|=|\tilde J_{(q,n, u)}^-|=(\lambda_1+\lambda_2)m.\end{equation}

\begin{lemma} \label{LemmaLL}
Let $q$ be even and $m \ge 5$ be odd. Then the following holds.
\begin{enumerate}
  \item $$\lambda_1=\begin{cases}
uq^{\bar m}/2-(u^2-u)q/4-u^2(q-1)/4, & \text{ if $u$ is even; } \\
uq^{\bar m}/2-(u^2-u)q/4-(u^2-1)(q-1)/4, & \text{ if $u$ is odd.} \end{cases}$$
  \item $$\lambda_2=\begin{cases}
((u^2-u)q-u^2)/4, & \text{ if $u$ is even; } \\
(u^2-1)(q-1)/4, & \text{ if $u$ is odd.} \end{cases}$$
\end{enumerate}
\end{lemma}

\begin{proof}
Notice that $q$ is even. It then follows from Proposition \ref{PO} and \eqref{Lambda1} that
$$\lambda_1=uq^{\bar m}/2-|\{\tilde{j} \in J_1: \tilde{j} \text{ is odd}\}|-|\{\tilde{j} \in J_2: \tilde{j} \text{ is odd}\}|.$$
By \eqref{J1} and \eqref{J2}, it is easy to see that
\begin{eqnarray*}|\{\tilde{j} \in J_1: \tilde{j} \text{ is odd}\}|&=&|\{j_{\bar m}q^{\bar m}+j_1q+j_0: 1 \le j_{\bar m} \le u-1, 0 \le j_1 < j_{\bar m}, 1 \le
\text{ odd } j_0 \le q-1\}| \\ &=&(u^2-u)q/4\end{eqnarray*} and
$$|\{\tilde{j} \in J_2: \tilde{j} \text{ is odd}\}|=\begin{cases}
u^2(q-1)/4, & \text{ if $u$ is even; } \\
(u^2-1)(q-1)/4, & \text{ if $u$ is odd.} \end{cases}$$
Then we prove the conclusion on $\lambda_1$.

Define $$\text{CL}_1=\{cl(\tilde{j}): \tilde{j} \in J_1\} \text{ and } \text{CL}_2=\{cl(\tilde{j}): \tilde{j} \in J_2\}.$$
By Proposition \ref{PO}, we have
$$\text{CL}_1=\{j_1q^{\bar m}+j_0q^{\bar m-1}+j_{\bar m}: 1 \le j_{\bar m} \le u-1, 0 \le j_1 < j_{\bar m}, 1 \le j_0 \le q-1\}$$ and
$$\text{CL}_2=\{j_0q^{\bar m}+j_{\bar m}q+j_{\bar m-1}: 1 \le j_{\bar m} \le u-1, 1 \le j_{\bar m-1} \le q-1, 1 \le j_0 \le j_{\bar m}\}.$$
It then follows from \eqref{Lambda2} that
$$\lambda_2=|\{\tilde{j} \in J_1: j_0 \text{ is odd and } j_{\bar m} \text{ is even}\}|+|\{\tilde{j} \in J_2: j_0 \text{ is odd and } j_{\bar m-1} \text{ is even}\}|.$$
One can easily check that
$$\lambda_2=\begin{cases}
((u^2-u)q-u^2)/4, & \text{ if $u$ is even; } \\
(u^2-1)(q-1)/4, & \text{ if $u$ is odd.} \end{cases}$$
This completes the proof.
\end{proof}

The following proposition follows from Lemma \ref{LemmaLL} and \eqref{LL} directly.

\begin{proposition} \label{propositionJPM} Let $m \ge 5$ be odd. Then
$$|\tilde J_{(q,n, u)}^+|=|\tilde J_{(q,n, u)}^-|=\begin{cases}
(uq^{\bar m}/2-u^2q/4)m, & \text{ if $u$ is even; } \\
\Big(uq^{\bar m}/2-(u^2-u)q/4\Big)m, & \text{ if $u$ is odd.} \end{cases}$$
\end{proposition}

\begin{theorem} \label{theoremQEMO} Let $m \ge 5$ be an odd integer, $q$ even, and $\delta=uq^{\frac {m+1} 2}/2+1$.
\begin{enumerate}
  \item If $u$ is even, then $\mathcal C_{(q,n, \delta, \frac {n+1} 2)}$ and $\mathcal C_{(q,n, \delta, \frac {n+1} 2-(\delta-1))}$ both have length
$n$, dimension $$q^m-1-(uq^{\frac {m+1} 2}/2-u^2q/4)m,$$ and minimum distance $d \ge \delta$.
  \item If $u$ is odd, then $\mathcal C_{(q,n, \delta, \frac {n+1} 2)}$ and $\mathcal C_{(q,n, \delta, \frac {n+1} 2-(\delta-1))}$ both have length
$n$, dimension $$k=q^m-1-\Big(uq^{\frac {m+1} 2}/2-(u^2-u)q/4\Big)m,$$ and minimum distance $d \ge \delta$.
\end{enumerate}
\end{theorem}

\begin{proof}
The desired conclusions follow from Proposition \ref{propositionJPM} and the BCH bound directly.
\end{proof}

\begin{example}
\begin{enumerate}
  \item When $(q, m, u)=(2, 7, 1)$ in the above theorem, the code $\mathcal C_{(q,n, \delta, \frac {n+1} 2)}$ or $(\mathcal C_{(q,n, \delta, \frac {n+1} 2-(\delta-1))})$ has parameters $[127, 71, 19]$, which are the best parameters for linear codes according to the Database.
  \item When $(q, m, u)=(4, 5, 1), (4, 5, 2), (4, 5, 3)$ in the above theorem, the code $\mathcal C_{(q,n, \delta, \frac {n+1} 2)}$ or $(\mathcal C_{(q,n, \delta, \frac {n+1} 2-(\delta-1))})$ has parameters $[1023, 863, d \ge 33]$, $[1023, 723, d \ge 65]$, and $[1023, 573, d \ge 97]$, respectively.
\end{enumerate}
\end{example}

The following conclusion will be employed to determine the dimension of the code $\codeodd$ when $m \ge 5$ is odd.

\begin{proposition} \label{PIEQOM}
For odd $m \ge 5$, we have
$$\tilde J_{(q,n,u)}^+ \cap \tilde J_{(q,n,u)}^-=\bigcup_{l \in \tilde {\mathcal J}_O}C_{\bar n+(l-1)/2} \cup C_{\bar n-(l+1)/2},$$
where the union is disjoint and
\begin{align*}
\tilde {\mathcal J}_O=\{&l_{\bar m}q^{\bar m}+l_{\bar m-1}q^{\bar m-1}+(q-1)(q^{\bar m-2}+\cdots+q^2+q)+l_0:0 \le l_{\bar m} \le u-1, \\
                        & 0 \le \text{ even } l_{\bar m-1} \le q-2, q-u \le \text{ odd } l_0 \le q-1\}.
\end{align*}
Moreover,
$$|\tilde J_{(q,n,u)}^+ \cap \tilde J_{(q,n,u)}^-|=\begin{cases}
\frac{u^2qm}{2}, & \text{ if $u$ is even;} \\
\frac{u(u+1)qm}{2}, &\text{ if $u$ is odd.}\end{cases}$$
\end{proposition}

\begin{proof} We are going to find the integers $i$ and $j$ with $1 \le i \le uq^{\bar m}$ and $1 \le j \le uq^{\bar m}$ such that
$$
C_{\bar n +i}=C_{\bar n -j}.
$$
This is equivalent to
$$(2i+1)+(2j-1) q^\ell \equiv 0 \pmod n$$ for some $1 \le \ell \le m-1$. Recall that in Remark~\ref{rem-1}, for $m \ge 5$ being odd, the integers $1 \le i_1,j_1 \le uq^{\bar m}$ satisfying
$$
i_1+j_1q^{\ell} \equiv 0 \pmod n
$$
have been characterized. Using this result, we can further characterize the odd integers $i_1$ and $i_2$ satisfying $i_1=2i+1$, $j_1=2j-1$ such that
$$
C_{\bar n +(i_1-1)/2}=C_{\bar n -(j_1+1)/2}.
$$
The remaining part of the theorem follows from Remark~\ref{rem-1} by employing a straightforward calculation.
\end{proof}

\begin{theorem} Let $m \ge 5$ be an odd integer, $q$ even, and $\delta=uq^{\frac {m+1} 2}/2+1$. Then $\codeodd$ has length
$n$, dimension $$k=q^m-1-(uq^{\frac {m+1} 2}-u^2q)m,$$ and minimum distance $d \ge 2\delta-1$.
\end{theorem}

\begin{proof}
The desired conclusion follows from Theorem \ref{theoremQEMO}, Proposition \ref{PIEQOM}, and the BCH bound.
\end{proof}

\begin{example}
\begin{enumerate}
  \item When $(q, m, u)=(2, 7, 1)$ in the above theorem, the code $\C_{(2,127,17,56)}$ has parameters $[127, 29, 37]$.
  \item When $(q, m, u)=(4, 5, 1), (4, 5, 2), (4, 5, 3)$ in the above theorem, the code $\codeodd$ has parameters $[1023, 723, d \ge 65]$, $[1023, 463, d \ge 129]$, and $[1023, 243, d \ge 193]$, respectively.
\end{enumerate}
\end{example}

\emph{B. Parameters of $\codeodd$ when $m$ is even}

It has been seen from Proposition \ref{PRQE} that $C_{\bar n+i} \ne C_{\bar n+j}$ if and only if
$C_{2i+1} \ne C_{2j+1}$ (resp. $C_{\bar n-i} \ne C_{\bar n-j}$ if and only if
$C_{2i-1} \ne C_{2j-1}$). Let $J$ be the set of integers that are not coset leaders, which are given by \eqref{JE}.
Note that $1 \le 2j+1 \le uq^{\bar m}-1$ if $0 \le j \le uq^{\bar m}/2-1$ and $1 \le 2j-1 \le uq^{\bar m}-1$ if $1 \le j \le uq^{\bar m}/2$. Using the same arguments at the beginning of previous subsection, we can see that
\begin{equation}\label{TTT} |\tilde J_{(q,n, u)}^+|=|\tilde J_{(q,n, u)}^-|=\theta_1 m+\theta_2 m/2+\theta_3 m,\end{equation}
where
\begin{equation} \label{Theta1} \theta_1=|\{1 \le  \tilde{j} \le uq^{\bar m}-1:
\tilde{j} \text{ is an odd coset leader and } |C_{\tilde j}|=m\}|, \end{equation}
\begin{equation} \label{Theta2} \theta_2=|\{1 \le  \tilde{j} \le uq^{\bar m}-1:
\tilde{j} \text{ is an odd coset leader and } |C_{\tilde j}|=m/2\}|, \end{equation}
and
\begin{equation} \label{Theta3} \theta_3=|\{1 \le \tilde{j} \le uq^{\bar m}-1:
\tilde{j} \in J  \text{ is odd, $cl(\tilde{l})$ is even}\}|. \end{equation}

\begin{lemma} \label{LemmaTTT}
Let $q$ be even and $m \ge 2$ be even. Then we have the following.
\begin{enumerate}
  \item $$\theta_1=\begin{cases}
uq^{\bar m}/2-u^2/4, & \text{ if $u$ is even; } \\
uq^{\bar m}/2-(u^2-1)/4, & \text{ if $u$ is odd.} \end{cases}$$
\item $$\theta_2=\begin{cases}
u/2, & \text{ if $u$ is even; } \\
(u-1)/2, & \text{ if $u$ is odd.} \end{cases}$$
  \item $$\theta_3=\begin{cases}
u(u-2)/8, & \text{ if $u$ is even; } \\
(u^2-1)/8, & \text{ if $u$ is odd.} \end{cases}$$
\end{enumerate}
\end{lemma}

\begin{proof}
Notice that $q$ is even. It then follows from Proposition \ref{PE}, \eqref{Theta1}, and \eqref{Theta2} that
$$\theta_1=uq^{\bar m}/2-\theta_2-|\{\tilde{j} \in J: \tilde{j} \text{ is odd}\}|.$$
By \eqref{JE}, it is easy to see that
\begin{eqnarray*}|\{\tilde{j} \in J: \tilde{j} \text{ is odd}\}|&=&|\{j_{\bar m}q^{\bar m}+j_0: 1 \le j_0 < j_{\bar m} \le u-1,
j_0 \text{ odd }\}| \\ &=&\begin{cases}
u(u-2)/4, & \text{ if $u$ is even, }\\
(u-1)^2/4, & \text{ if $u$ is odd. } \end{cases} \end{eqnarray*}
In addition, it follows from Proposition \ref{PE} that
$$\theta_2=|\{v(q^{\bar m}+1): 1 \le \text{ odd } v \le u-1|=\begin{cases}
u/2, & \text{ if $u$ is even; } \\
(u-1)/2, & \text{ if $u$ is odd.} \end{cases}$$
Then we get the conclusions on $\theta_1$ and $\theta_2$.

Define
$$
\text{CL}=\{cl(\tilde j) \mid \tilde j \in J \}.
$$
It follows from Proposition \ref{PE} that
$$\text{CL}=\{j_0q^{\bar m}+j_{\bar m}: j_0+1 \le j_{\bar m} \le u-1, 1 \le j_0 \le u-1\}.$$
Then we can deduce from \eqref{Theta3} that
$$\theta_3=|\{j_0q^{\bar m}+j_{\bar m}: 1 \le j_0 <  j_{\bar m} \le u-1, j_0 \text{ odd, } j_{\bar m} \text{ even}\}|.$$
It can be easily verified that $$\theta_3=\begin{cases}
u(u-2)/8, & \text{ if $u$ is even; } \\
(u^2-1)/8, & \text{ if $u$ is odd.} \end{cases}$$
This completes the proof.
\end{proof}

The following results follow from Lemma \ref{LemmaTTT} and \eqref{TTT} directly.

\begin{proposition} \label{propositionJPMME} Let $m \ge 2$ be even. Then
$$|\tilde J_{(q,n, u)}^+|=|\tilde J_{(q,n, u)}^-|=\begin{cases}
\Big(uq^{\bar m}-u^2/4\Big)\frac m 2, & \text{ if $u$ is even; } \\
\Big(uq^{\bar m}-(u-1)^2/4\Big)\frac m 2, & \text{ if $u$ is odd.} \end{cases}$$
\end{proposition}

\begin{theorem} \label{theoremQEME} Let $m \ge 2$ be an even integer, $q$ even, and $\delta=uq^{\frac {m} 2}/2+1$.
\begin{enumerate}
  \item If $u$ is even, then $\mathcal C_{(q,n, \delta, \frac {n+1} 2)}$ and $\mathcal C_{(q,n, \delta, \frac {n+1} 2-(\delta-1))}$ both have length
$n$, dimension $$k=q^m-1-\Big(uq^{\frac m 2}-u^2/4\Big)\frac m 2,$$ and minimum distance $d \ge \delta$.
  \item If $u$ is odd, then $\mathcal C_{(q,n, \delta, \frac {n+1} 2)}$ and $\mathcal C_{(q,n, \delta, \frac {n+1} 2-(\delta-1))}$ both have lenth
$n$, dimension $$k=q^m-1-\Big(uq^{\frac m 2}-(u-1)^2/4\Big)\frac m 2,$$ and minimum distance $d \ge \delta$.
\end{enumerate}
\end{theorem}

\begin{proof}
The desired conclusions follow from Proposition \ref{propositionJPMME} and the BCH bound directly.
\end{proof}

\begin{example}
\begin{enumerate}
  \item When $(q, m, u)=(2, 6, 1)$ in the above theorem, the code $\mathcal C_{(q,n, \delta, \frac {n+1} 2)}$ or $(\mathcal C_{(q,n, \delta, \frac {n+1} 2-(\delta-1))})$ has parameters $[63, 39, 9]$, which are the best parameters for linear codes according to the Database and the best possible
      cyclic codes according to {\rm\cite[p. 260]{Dingbk15}}.
  \item When $(q, m, u)=(4, 4, 1), (4, 4, 2), (4, 4, 3)$, the code $\mathcal C_{(q,n, \delta, \frac {n+1} 2)}$ or $(\mathcal C_{(q,n, \delta, \frac {n+1} 2-(\delta-1))})$ has parameters $[255, 223, d \ge 9]$, $[255, 193, d \ge 17]$, and $[255, 161, d \ge 25]$, respectively.
\end{enumerate}
\end{example}

The following conclusion will be employed to investigate the parameters of the code $\codeodd$ when $m \ge 2$ is even.

\begin{proposition} \label{PIEQEM}
Let $m \ge 2$ be an even integer and $q$ be even. Suppose
$$
\begin{cases}
  1 \le u \le \frac{q}{2} & \mbox{if $m=2$,} \\
  1 \le u \le q-1 & \mbox{if $m \ge 4$.}
\end{cases}
$$
Then we have
$$\tilde J_{(q,n,u)}^+ \cap \tilde J_{(q,n,u)}^-=\bigcup_{l \in \tilde {\mathcal J}_E}C_{\bar n-(l+1)/2},$$
where the union is disjoint and
$$\tilde{{\mathcal J}}_E=\{j_{\bar m}q^{\bar m}+(q-1)(q^{\bar m-1}+q^{\bar m-2}+\cdots+q)+j_0: 0 \le \text{ even } j_{\bar m} \le u-1 \text{ and } q-u \le \text{ odd } j_0 \le q-1\}.$$
Moreover,
$$|\tilde J_{(q,n,u)}^+ \cap \tilde J_{(q,n,u)}^-|=\begin{cases}
u^2m/4, & \text{ if $u$ is even;} \\
(u+1)^2m/4 , &\text{ if $u$ is odd.}\end{cases}$$
\end{proposition}

\begin{proof}
We are going to find all the integers $i$ and $j$ with $1 \le i \le uq^{\bar m}$ and $1 \le j \le uq^{\bar m}$ such that
$$
C_{\bar n +i}=C_{\bar n -j}.
$$
This is equivalent to
$$(2i+1)+(2j-1) q^\ell \equiv 0 \pmod n$$ for some $1 \le \ell \le m-1$. Recall that in Remark~\ref{rem-2}, for $m \ge 2$ being even, the integers $1 \le i_1,j_1 \le uq^{\bar m}$ satisfying
$$
i_1+j_1q^{\ell} \equiv 0 \pmod n
$$
have been characterized. Using this result, we can further characterize the odd integers $i_1$ and $i_2$ satisfying $i_1=2i+1$, $j_1=2j-1$ such that
$$
C_{\bar n +(i_1-1)/2}=C_{\bar n -(j_1+1)/2}.
$$
As a consequence, we have $\tilde J_{(q,n,u)}^+ \cap \tilde J_{(q,n,u)}^-=\bigcup_{l \in  \tilde{\mathcal J}_E} C_{\bar n-(l+1)/2}$. By Proposition~\ref{PE}, each $l \in \tilde{\mathcal J}_E $ is a coset leader and $|C_{\bar n-l}|=m$. In particular, when $m=2$, we need $1 \le u \le \frac{q}{2}$ to ensure that each $l \in \tilde{\mathcal J}_E$ is a coset leader and $|C_{\bar n-l}|=m$. The remaining part of the theorem follows from Remark~\ref{rem-2} by employing a straightforward calculation.
\end{proof}

\begin{theorem} \label{TheoremCQEME} Let $m \ge 2$ be an even integer, $q$ even, and $\delta=uq^{\frac m 2}/2+1$. Suppose
$$
\begin{cases}
  1 \le u \le \frac{q}{2} & \mbox{if $m=2$} \\
  1 \le u \le q-1 & \mbox{if $m \ge 4$}
\end{cases}
$$
\begin{enumerate}
  \item If $u$ is even, then $\codeodd$ has length
$n$, dimension $$k=q^m-1-(uq^{\frac m 2}-u^2/2)m,$$ and minimum distance $d \ge 2 \delta-1$.
  \item If $u$ is odd, then $\codeodd$ has length
$n$, dimension $$k=q^m-1-(uq^{\frac m 2}-(u^2+1)/2)m,$$ and minimum distance $d \ge 2 \delta-1$.
\end{enumerate}
\end{theorem}

\begin{proof}
The desired conclusion follows from Theorem \ref{theoremQEME}, Proposition \ref{PIEQEM}, and the BCH bound.
\end{proof}

\begin{example}
\begin{enumerate}
  \item When $(q, m, u)=(2, 4, 1)$ in the above theorem, the code $C_{(2,15,5,6)}$ has parameters $[15, 3, 5]$, which are the best possible parameters for cyclic codes \cite[pp. 247]{Dingbk15}.
  \item When $(q, m, u)=(4, 4, 1), (4, 4, 2), (4, 4, 3)$ in the above theorem, the code $\codeodd$ has parameters $[255, 195, 17]$, $[255, 135, d \ge 33]$, and $[255, 83, d \ge 49]$, respectively.
\end{enumerate}
\end{example}

\begin{corollary}
When $u=1$ and $\delta=q^{\frac m 2}/2+1$, the true minimum distance of the code $\codeodd$ presented in
Theorem \ref{TheoremCQEME} is equal to $2\delta-1$.
\end{corollary}

\begin{proof}
Note that $b=\frac {n+1} 2-\delta+1$. It is easy to check that $(2 \delta-1) \mid \gcd(n, b-1)$ in this case. The desired result then
follows from Corollary \ref{corollarygener}.
\end{proof}

\emph{C. Parameters of $\codeodd$ with designed distance $q^t-1$, where $1 \le t \le \bar m$}

When $q$ is even, the parameters of the LCD code $\codeodd$ are described in the following theorem if $\codeodd$ has designed distance $2\delta-1=q^t-1$ for an integer $t$ with $1 \le t \le \bar m$.

\begin{theorem} \label{theoremLCDDDQTMO2} Let $q$ be even, $m \ge 2$ and $m \ne 3$.
Suppose $\codeodd$ has designed distance $2\delta-1=q^t-1$, where $1 \le t \le \bar m$. Then $\codeodd$ has length
$n$, dimension $$k=\begin{cases} q^m-1-(q^{\frac {m+1} 2}-q)m & \text{ if } m \ge 5 \text{ is odd and } t=\frac {m+1} 2, \\
q^m-1-(q^t-2)m & \text{ otherwise}, \end{cases}$$ and minimum distance $d \ge q^t-1$.
\end{theorem}

\begin{proof}
Set $\delta=\frac {q^t} 2$. Recall that the generator polynomial of the code $\codeodd$ is $g_{(q, n, 2\delta-1, \frac {n+1} 2-(\delta-1))}(x)$. By Lemma~\ref{AKSYF}, we have
$$\deg(g_{(q, n, 2\delta-1, \frac {n+1} 2-(\delta-1))}(x))=
(q^t-2)m-\Big|\Big(\bigcup_{1 \le j \le \delta-1}C_{\bar n+j}\Big) \bigcap \Big(\bigcup_{0 \le j \le \delta-2}C_{\bar n-j}\Big)\Big|.$$

When $m \ge 5$ is odd, the integers $1 \le i_1,j_1 \le uq^{\bar m}$ satisfying
$$
i_1+j_1q^{\ell} \equiv 0 \pmod n
$$
have been characterized in Remark~\ref{rem-1}. Using this result, we can show that
$$\Big(\bigcup_{0 \le j \le \delta-2}C_{\bar n+j}\Big) \bigcap \Big(\bigcup_{1 \le j \le \delta-1}C_{\bar n-j}\Big)=\begin{cases}
 \bigcup_{j \in J'}(C_{\bar n+(j-1)/2} \cup C_{\bar n-(j+1)/2}) & \text{ if } t=\frac {m+1} 2, \\
\emptyset & \text{ if } 1 \le t \le \frac {m-1} 2, \end{cases}$$
where $J'=\{j_{\bar m-1}q^{\bar m-1}+q^{\bar m-1}-1: 2 \le \text{ even } j_{\bar m-1} \le q-2\}$.

When $m \ge 2$ is even, by Proposition 29, for $ 1 \le t \le \frac m 2$, we have
$$\Big(\bigcup_{0 \le j \le \delta-2}C_{\bar n+j}\Big) \bigcap \Big(\bigcup_{1 \le j \le \delta-1}C_{\bar n-j}\Big)=\emptyset.$$

Therefore, we have
$$\Big|\Big(\bigcup_{0 \le j \le \delta-2}C_{\bar n+j}\Big) \bigcap \Big(\bigcup_{1 \le j \le \delta-1}C_{\bar n-j}\Big)\Big|=\begin{cases} (\frac q 2-1)2m & \text{ if } m \ge 5 \text{ is odd and } t=\frac {m+1} 2, \\
0 & \text{ otherwise}. \end{cases}$$
Thus, the dimension is obtained. Moreover, the minimum distance $d \ge q^t-1$ follows from the BCH bound.
\end{proof}

We remark that the minimum distance of the code $\codeodd$ given in Theorem \ref{theoremLCDDDQTMO2} may be larger than $q^t-1$.

\begin{example}
\begin{enumerate}
  \item When $(q, m, t)=(2, 7, 2),(2, 7, 3),(2, 7, 4)$ in the above theorem, $\codeodd$ has parameters $[127, 113, 5]$, $[127, 85, 11]$ and $[127, 29, 37]$ with designed distance $3$, $7$, and $15$, respectively.
  \item When $(q, m, t)=(2, 6, 2), (2, 6, 3)$ in the above theorem, the code $\codeodd$ has parameters $[63, 51, 3]$ and $[63, 27, 7]$.
\end{enumerate}
\end{example}

\section{Parameters of LCD BCH code $\mathcal C_{(q, n, 2\delta, n-\delta+1)}$}

In this section, we investigate the parameters of the LCD BCH code $\mathcal C_{(q, n, 2\delta, n-\delta+1)}$.

\subsection{The dimension of $\ocode$ when $\de$ is relatively small}

Every positive integer $s$ with $0 \le s \le n$ has a unique $q$-ary expansion as $s=\sum_{i=0}^{m-1}s_iq^i$, where $0 \le s_i \le q-1$. The $q$-ary expansion sequence of $s=\sum_{i=0}^{m-1}s_iq^i$ is denoted by $\ol{s}=(s_{m-1},s_{m-2},\ldots,s_0)$. Below, we simply call the $q$-ary expansion sequence of $s$ as the sequence of $s$, whenever this causes no confusion. The weight of $\ol{s}$ is defined to be the number of nonzero entries among $\ol{s}$ and denoted by $wt(\ol{s})$. Define the support of $\ol{s}$ as
$$
\supp(\ol{s})=\{0 \le i \le m-1 \mid s_i \ne 0\}.
$$

\begin{lemma}\label{lem-cyccoset2}
Let $m \ge 2$. Then the following holds.
\begin{itemize}
\item[1)] When $m$ is odd, for $1 \le i,j \le q^{(m+1)/2}$, $-j \in C_i$ if and only if
$$
\ol{i}=(0,\ldots,0,\underset{\frac{m-1}{2}}{q-1},\ldots,\underset{1}{q-1},u),\quad \ol{j}=(0,\ldots,0,\underset{\frac{m-1}{2}}{q-1-u},\underset{\frac{m-3}{2}}{q-1},\ldots,q-1),
$$
or
$$
\ol{i}=(0,\ldots,0,\underset{\frac{m-1}{2}}{q-1-u},\underset{\frac{m-3}{2}}{q-1},\ldots,q-1), \quad
\ol{j}=(0,\ldots,0,\underset{\frac{m-1}{2}}{q-1},\ldots,\underset{1}{q-1},u),
$$
where $0 \le u \le q-1$.

\item[2)] When $m$ is even and $q>2$, for $1 \le i,j \le 2q^{m/2}$, $-j \in C_i$ if and only if
$$
\ol{i}=(0,\ldots,0,\underset{\frac{m}{2}}{1},\underset{\frac{m}{2}-1}{q-1},\ldots,q-1,q-2), \quad
\ol{j}=(0,\ldots,0,\underset{\frac{m}{2}}{1},\underset{\frac{m}{2}-1}{q-1},\ldots,q-1,q-2),
$$
or
$$
\ol{i}=(0,\ldots,0,\underset{\frac{m}{2}-1}{q-1},\ldots,\underset{1}{q-1},q-2), \quad
\ol{j}=(0,\ldots,0,\underset{\frac{m}{2}}{1},\underset{\frac{m}{2}-1}{q-1},\ldots,q-1),
$$
or
$$
\ol{i}=(0,\ldots,0,\underset{\frac{m}{2}}{1},\underset{\frac{m}{2}-1}{q-1},\ldots,q-1), \quad
\ol{j}=(0,\ldots,0,\underset{\frac{m}{2}-1}{q-1},\ldots,\underset{1}{q-1},q-2),
$$
or
$$
\ol{i}=(0,\ldots,0,\underset{\frac{m}{2}-1}{q-1},\ldots,q-1), \quad
\ol{j}=(0,\ldots,0,\underset{\frac{m}{2}-1}{q-1},\ldots,q-1).
$$
\item[3)] When $m$ is even and $q=2$, for $1 \le i,j \le 2^{(m/2)+1}$, $-j \in C_i$ if and only if
$$
\ol{i}=(0,\ldots,0,\underset{\frac{m}{2}-2}{1},\ldots,1), \quad
\ol{j}=(0,\ldots,0,\underset{\frac{m}{2}}{1},\ldots,1),
$$
or
$$
\ol{i}=(0,\ldots,0,\underset{\frac{m}{2}}{1},\ldots,1), \quad
\ol{j}=(0,\ldots,0,\underset{\frac{m}{2}-2}{1},\ldots,1),
$$
or
$$
\ol{i}=(0,\ldots,0,\underset{\frac{m}{2}-1}{1},\ldots,1), \quad
\ol{j}=(0,\ldots,0,\underset{\frac{m}{2}-1}{1},\ldots,1),
$$
or
$$
\ol{i}=(0,\ldots,0,\underset{\frac{m}{2}}{1},0,1,\ldots,1), \quad
\ol{j}=(0,\ldots,0,\underset{\frac{m}{2}}{1},\ldots,1,0,1),
$$
or
$$
\ol{i}=(0,\ldots,0,\underset{\frac{m}{2}}{1},\ldots,1,0,1), \quad
\ol{j}=(0,\ldots,0,\underset{\frac{m}{2}}{1},0,1,\ldots,1).
$$
\end{itemize}
\end{lemma}
\begin{proof}
1) If $-j \in C_i$, then there exists an $l$ with $0 \le l \le m-1$, such that $q^li+j\equiv 0\pmod{n}$. Hence,
$$
\ol{q^li+j}=(q-1,q-1,\ldots,q-1).
$$
Since $m=wt(\ol{q^li+j})\le wt(\ol{i})+wt(\ol{j}) \le m+1$, we have
$\{wt(\ol{i}),wt(\ol{j})\}=\{\frac{m-1}{2},\frac{m+1}{2}\}$ or $\{wt(\ol{i}),wt(\ol{j})\}=\{\frac{m+1}{2}\}.$
If $\{wt(\ol{i}),wt(\ol{j})\}=\{\frac{m-1}{2},\frac{m+1}{2}\}$, then clearly, $\supp(\ol{q^li}) \cap \supp(\ol{j})=\emptyset$. Otherwise, if $\supp(\ol{q^li}) \cap \supp(\ol{j}) \ne \emptyset$, there is at least one entry in $\ol{q^li+j}$, which is not $q-1$. Hence, $\ol{q^li}$ and $\ol{j}$ must have the following two forms
$$
\ol{q^li}=(\underset{m-1}{q-1},\ldots,\underset{\frac{m+1}{2}}{q-1},0,\ldots,0),\quad
\ol{j}=(0,\ldots,0,\underset{\frac{m-1}{2}}{q-1},\ldots,q-1),
$$
or
$$
\ol{q^li}=(\underset{m-1}{q-1},\ldots,\underset{\frac{m-1}{2}}{q-1},0,\ldots,0),\quad
\ol{j}=(0,\ldots,0,\underset{\frac{m-3}{2}}{q-1},\ldots,q-1),
$$
If $\{wt(\ol{i}),wt(\ol{j})\}=\{\frac{m+1}{2}\}$, then clearly, $|\supp(\ol{q^li}) \cap \supp(\ol{j})|\ge 1$. If $|\supp(\ol{q^li}) \cap \supp(\ol{j})|>1$, then there is at least one entry in $\ol{q^li+j}$, which is not $q-1$. Hence, $|\supp(\ol{q^li}) \cap \supp(\ol{j})|=1$. Therefore, $\ol{q^li}$ and $\ol{j}$ must have the following $2q-4$ forms
$$
\ol{q^li}=(\underset{m-1}{q-1},\ldots,\underset{\frac{m+1}{2}}{q-1},u,0,\ldots,0), \quad
\ol{j}=(0,\ldots,0,\underset{\frac{m-1}{2}}{q-1-u},\underset{\frac{m-3}{2}}{q-1},\ldots,q-1),
$$
or
$$
\ol{q^li}=(\underset{m-1}{q-1},\ldots,\underset{\frac{m+1}{2}}{q-1},\underset{\frac{m-1}{2}}{0},\ldots,0,u), \quad
\ol{j}=(0,\ldots,0,\underset{\frac{m-1}{2}}{q-1},\ldots,q-1,q-1-u),
$$
where $1 \le u \le q-2$. Therefore, the conclusion follows.

2) If $-j \in C_i$, then there exists an $l$ with $0 \le l \le m-1$, such that $q^li+j\equiv 0\pmod{n}$. Hence,
$$
\ol{q^li+j}=(q-1,q-1,\ldots,q-1).
$$
Since $m=wt(\ol{q^li+j})\le wt(\ol{i})+wt(\ol{j}) \le m+2$, we must have
$$
\{wt(\ol{i}),wt(\ol{j})\}=\begin{cases}
\{\frac{m}{2}+1\}, \mbox{ or} \\
\{\frac{m}{2},\frac{m}{2}+1\}, \mbox{ or} \\
\{\frac{m}{2}-1,\frac{m}{2}+1\}, \mbox{ or} \\
\{\frac{m}{2}\}.
\end{cases}
$$
If $\{wt(\ol{i}),wt(\ol{j})\}=\{\frac{m}{2}+1\}$, then $|\supp(\ol{q^li}) \cap \supp(\ol{j})|=2$. Hence, $\ol{q^li}$ and $\ol{j}$ must have the following form
$$
\ol{q^li}=(\underset{m-1}{q-1},\ldots,\underset{\frac{m}{2}+1}{q-1},\underset{\frac{m}{2}}{q-2},\underset{\frac{m}{2}-1}{0},\ldots,0,1),\quad
\ol{j}=(0,\ldots,0,\underset{\frac{m}{2}}{1},\underset{\frac{m}{2}-1}{q-1},\ldots,q-1,q-2).
$$
If $\{wt(\ol{i}),wt(\ol{j})\}=\{\frac{m}{2}, \frac{m}{2}+1\}$, then $|\supp(\ol{q^li}) \cap \supp(\ol{j})|=1$. Hence, $\ol{q^li}$ and $\ol{j}$ must have the following two forms
$$
\ol{q^li}=(\underset{m-1}{q-1},\ldots,\underset{\frac{m}{2}+1}{q-1},\underset{\frac{m}{2}}{q-2},0,\ldots,0), \quad
\ol{j}=(0,\ldots,0,\underset{\frac{m}{2}}{1},\underset{\frac{m}{2}-1}{q-1},\ldots,q-1),
$$
or
$$
\ol{q^li}=(\underset{m-1}{q-1},\ldots,\underset{\frac{m}{2}}{q-1},\underset{\frac{m}{2}-1}{0},\ldots,0,1), \quad
\ol{j}=(0,\ldots,0,\underset{\frac{m}{2}-1}{q-1},\ldots,\underset{1}{q-1},q-2).
$$
If $\{wt(\ol{i}),wt(\ol{j})\}=\{\frac{m}{2}-1, \frac{m}{2}+1\}$, then $|\supp(\ol{q^li}) \cap \supp(\ol{j})|=0$. Hence, there is at least one entry in $\ol{(q^li+j)}$ which is not equal to $q-1$.
If $\{wt(\ol{i}),wt(\ol{j})\}=\{\frac{m}{2}\}$, then $|\supp(\ol{q^li}) \cap \supp(\ol{j})|=0$. Hence, $\ol{q^li}$ and $\ol{j}$ must have the following form
$$
\ol{i}=(\underset{m-1}{q-1},\ldots,\underset{\frac{m}{2}}{q-1},0,\ldots,0), \quad
\ol{j}=(0,\ldots,0,\underset{\frac{m}{2}-1}{q-1},\ldots,q-1).
$$
Therefore, the conclusion follows.

3) The proof is similar to that of 2) and is omitted here.
\end{proof}

As a consequence, we have the following proposition.

\begin{proposition}\label{prop-cyccoset}
Let $m \ge 2$.
\begin{itemize}
\item[1)] Suppose $m$ is odd. Then
\begin{align*}
&|\{(cl(i),cl(j)) \mid -j \in C_i, 1\le i,j\le l\}| \\
&= \begin{cases}
0 & \mbox{if $1 \le l \le q^{(m+1)/2}-q$}, \\
2h & \mbox{if $l=q^{(m+1)/2}-q+h$, $1 \le h \le q-2$}, \\
2(q-1) & \mbox{if $q^{(m+1)/2}-1\le l\le q^{(m+1)/2}$}.
\end{cases}
\end{align*}
\item[2)] Suppose $m$ is even and $q>2$. Then
\begin{align*}
&|\{(cl(i),cl(j)) \mid -j \in C_i, 1\le i,j\le l\}| \\
&= \begin{cases}
0 & \mbox{if $1 \le l \le q^{m/2}-2$}, \\
1 & \mbox{if  $q^{m/2}-1\le l\le 2q^{m/2}-3$}, \\
2 & \mbox{if $l=2q^{m/2}-2$}, \\
4 & \mbox{if $2q^{m/2}-1 \le l \le 2q^{m/2}$}.
\end{cases}
\end{align*}
\item[3)] Suppose $m \ge 4$ is even and $q=2$. Then
\begin{align*}
&|\{(cl(i),cl(j)) \mid -j \in C_i, 1\le i,j\le l\}| \\
&= \begin{cases}
0 & \mbox{if $1 \le l \le 2^{m/2}-2$}, \\
1 & \mbox{if $2^{m/2}-1 \le l \le 2^{(m/2)+1}-4$}, \\
3 & \mbox{if $2^{(m/2)+1}-3 \le l \le 2^{(m/2)+1}-2$}, \\
5 & \mbox{if $2^{(m/2)+1}-1 \le l \le 2^{(m/2)+1}$}.
\end{cases}
\end{align*}
\end{itemize}
\end{proposition}

Combining Lemma~\ref{AKSYF} and Proposition~\ref{prop-cyccoset}, we have the following theorem.

\begin{theorem}\label{thm-gene}
Let $m \ge 2$. Let $\de$ be an integer satisfying
$$
\begin{cases}
2 \le \de \le q^{(m+1)/2}+1 & \mbox{if $m$ is odd,} \\
2 \le \de \le 2q^{m/2}+1 & \mbox{if $m$ is even.}
\end{cases}
$$
Let $\delta_q$ and $\delta_0$ be the unique integers such that $\delta-1 =\delta_q q+\delta_0$, where $0\le\delta_0<q$. Then $\ocode$ has parameters $[q^m-1,k,d \ge 2\de]$, in which the dimension $k$ is given below.
\begin{itemize}
\item[1)] When $m$ is odd,
$$
k=\begin{cases}
q^m-2-2m(\de_q(q-1)+\de_0) & \mbox{if $\de\le q^{(m+1)/2}-q$,} \\
q^m-2-2m(q^{(m-1)/2}-1)(q-1) & \mbox{if $q^{(m+1)/2}-q+1\le\de\le q^{(m+1)/2}+1$.}
\end{cases}
$$
\item[2)] When $m$ is even and $q>2$,
$$
k=\begin{cases}
q^m-2-2m(\de_q(q-1)+\de_0) & \mbox{if $\de\le q^{m/2}-1$,} \\
q^m-2-2m(\de_q(q-1)+\de_0-\frac{1}{2}) & \mbox{if $q^{m/2}\le\de\le q^{m/2}+1$,} \\
q^m-2-2m(\de_q(q-1)+\de_0-1) & \mbox{if $q^{m/2}+2\le\de\le 2q^{m/2}-2$,} \\
q^m-2-2m(\de_q(q-1)+\de_0-\frac{3}{2}) & \mbox{if $\de=2q^{m/2}-1$,} \\
q^m-2-2m(\de_q(q-1)+\de_0-\frac{5}{2}) & \mbox{if $2q^{m/2}\le\de\le 2q^{m/2}+1$.}
\end{cases}
$$
\item[3)] When $m \ge 4$ is even and $q=2$,
$$
k=\begin{cases}
2^m-2-2m(\de_q+\de_0) & \mbox{if $\de\le 2^{m/2}-1$ and $m \ge 4$,} \\
2^m-2-2m(\de_q+\de_0-\frac{1}{2}) & \mbox{if $2^{m/2}\le\de\le 2^{m/2}+1$ and $m \ge 4$,} \\
2^m-2-2m(\de_q+\de_0-1) & \mbox{if $2^{m/2}+2\le\de\le2^{(m/2)+1}-3$ and $m \ge 4$,} \\
2^m-2-2m(\de_q+\de_0-2) & \mbox{if $2^{(m/2)+1}-2\le\de\le2^{(m/2)+1}-1$ and $m \ge 6$,} \\
2^m-2-2m(\de_q+\de_0-3) & \mbox{if $2^{(m/2)+1}\le\de\le 2^{(m/2)+1}+1$ and $m \ge 6$.}
\end{cases}
$$
\end{itemize}
In addition, the minimum distance $d$ of the code satisfies that $d\ge2\de$.
\end{theorem}
\begin{proof}
Let $\overline{g}_{(q,m,\delta)}(x)$ be the generator polynomial of $\ocode$. For the dimension of the code, we only prove 2) since the proofs of 1) and 3) are similar. By 2) of Lemma~\ref{AKSYF}, the degree of $\ol{g}_{(q,m,\delta)}(x)$ equals
$$
1+2m(\de_q(q-1)+\de_0)-\ep m-|\{(cl(i),cl(j))\mid -j\in C_i, 1 \le i,j\le \de-1\}|m,
$$
where
$$
\ep=\begin{cases}
0 & \mbox{if $\de \le q^{m/2}+1$}, \\
1 & \mbox{if $\de \ge q^{m/2}+2$}.
\end{cases}
$$
With this conclusion on the degree of the generator polynomial and Proposition~\ref{prop-cyccoset}, we have
$$
\deg{(\ol{g}_{(q,m,\delta)}(x))}=\begin{cases}
1+2m(\de_q(q-1)+\de_0) & \mbox{if $\de\le q^{m/2}-1$,} \\
1+2m(\de_q(q-1)+\de_0-\frac{1}{2}) & \mbox{if $q^{m/2}\le\de\le q^{m/2}+1$,} \\
1+2m(\de_q(q-1)+\de_0-1) & \mbox{if $q^{m/2}+2\le\de\le 2q^{m/2}-2$,} \\
1+2m(\de_q(q-1)+\de_0-\frac{3}{2}) & \mbox{if $\de=2q^{m/2}-1$,} \\
1+2m(\de_q(q-1)+\de_0-\frac{5}{2}) & \mbox{if $2q^{m/2}\le\de\le 2q^{m/2}+1$.}
\end{cases}
$$
Therefore, the conclusion on the dimension in 2) follows. Moreover, by the BCH bound, $\ocode$ has minimum distance $d \ge 2\de$.
\end{proof}

\begin{remark}
For the code $\ocode$, if
\begin{equation}\label{condi}
\sum_{i=0}^{\de} \binom{n}{i}(q-1)^i > q^{n-k},
\end{equation}
then $d \le 2\de$ by the sphere packing bound. Therefore, the knowledge on the dimension of the code $\ocode$ may provide more precise information on the minimum distance in some cases. As an illustration, we use Theorem~\ref{thm-gene} and the inequality (\ref{condi}) to get some binary codes $\C_{(2,n,\de,n-\de+1)}$ with $d=2\de$, which are listed in Table~\ref{tab-mindis}. Note that the codes listed in Table~\ref{tab-mindis} is optimal in the sense that given the length and dimension, the minimum distance is the largest possible. According to Inequality (\ref{condi}), the increasing of their minimum distances is impossible due to the sphere packing bound.
\end{remark}

\begin{table}[h]
\begin{center}
\caption{Some optimal binary code $\C_{(2,n,\de,n-\de+1)}$ with $d=2\de$}\label{tab-mindis}
\begin{tabular}{|c|c|} \hline
$m$  &   $\de$  \\ \hline
\{5,6,7\} & \{3\} \\
\{8,9,10,11,12,13\} & \{3,5\} \\
\{14,15,17,17,18,19\} & \{3,5,7\} \\
\{20\} & \{3,5,7,9\} \\ \hline
\end{tabular}
\end{center}
\end{table}

\begin{remark}
Theorem~\ref{thm-gene} gives the dimension of $\ocode$ when $\de$ is relatively small, in which $\de$ is approximately the square root of the length $n$. In this case, the size of each cyclotomic coset containing $i$, where $-\de \le i \le \de$, follows form Lemma~\ref{AKSYF}. Moreover, Lemma~\ref{lem-cyccoset2} characterizes all $1 \le i,j \le \de$ satifying $-j \in C_i$. For a larger $\de$, the size of cyclotomic cosets, as well as the cases in which $-j \in C_i$, become much more complicated. Hence, from this viewpoint, it is difficult to extend the result of Theorem~\ref{thm-gene} to a larger $\de$.
\end{remark}

\begin{remark}
Since when $q$ is odd, $\codeeven$ and $\ocode$ are monomially equivalent, Theorem~\ref{thm-gene} also gives the dimension of $\codeeven$ for $2 \le \de \le q^{(m+1)/2}$ when $m$ is odd and for $2 \le \de \le 2q^{m/2}+1$ when $m$ is even. Moreover, Theorem~\ref{theoremLCDDDQTMO} is also a direct consequence of Theorem~\ref{thm-gene}.
\end{remark}

Due to the equivalence between $\codeeven$ and $\ocode$ when $q$ is even, the following theorem follows immediately from Theorems~\ref{theoremLCDQOMO} and \ref{TheoremCQOME}.

\begin{theorem} Let $q$ be odd and $\delta=uq^{\bar m}+1$, where
$$
\begin{cases}
  1 \le u \le \frac{q-1}{2} & \mbox{if $m=2$,} \\
  1 \le u \le q-1 & \mbox{if $m \ge 4$.}
\end{cases}
$$
\begin{enumerate}
  \item When $m \ge 5$ is an odd integer, the code $\mathcal C_{(q, n, 2\delta, n-\delta+1)}$ has length
$n$, dimension $$k=q^m-2-2(uq^{\frac {m-1} 2}-2u^2+u)(q-1)m,$$ and minimum distance $d \ge 2\delta$.
  \item When $m \ge 2$ is an even integer, the code $\mathcal C_{(q, n, 2\delta, n-\delta+1)}$  has length
$n$, dimension $$k=q^m-2-2uq^{\frac m 2-1}(q-1)m+(2u^2-2u+1)m,$$ and minimum distance  $d \ge 2\delta$.
\end{enumerate}
\end{theorem}

\subsection{The dimension of $\ocode$ when $\de=q^\lam$ and $\frac{m}{2}\le\lam\le m-1$}

In \cite{Mann}, the dimension of the narrow-sense primitive BCH code $\codenp$ with $\de=q^\lam$ was considered. The author derived two closed formulas concerning the dimension of such code. In this subsection, we use the idea in \cite{Mann} to give an estimate of the dimension of the LCD BCH code $\ocode$ with $\de=q^\lam$, where $\frac{m}{2}\le\lam\le m-1$.

Let $s$ and $r$ be two positive integers. Given a sequence of length $s$ and a fixed integer $a$ with $0 \le a \le q-1$, we say that the sequence contains a straight run of length $r$ with respect to $a$, if it has $r$ consecutive entries formed by $a$. If we view the sequence as a circle where the first and last entry are glued together, we say that the sequence contains a circular run of length $r$ with respect to $a$, if this circle has $r$ consecutive entries formed by $a$. When the specific choice of the integer $a$ does not matter, we simply say that the sequence has a straight or circular run of length $r$. Clearly, a straight run is also a circular run but the converse is not necessarily true. We use $l_r(s)$ to denote the number of sequences of length $s$, which contains a straight run of length $r$. Particularly, we define $l_r(0)=0$. The following is a recursive formula of $l_r(s)$ which was presented in \cite{Mann}.

\begin{result}{\rm \cite[p. 155]{Mann}}
Let $s$ and $r$ be two nonnegative integers. Then
$$
l_r(s)=\begin{cases}
  \begin{matrix}
  0 & \mbox{if $0\le s<r$}, \\
  1 & \mbox{if $s=r$},\\
  ql_r(s-1)+(q-1)(q^{s-r-1}-l_r(s-r-1)) & \mbox{if $s>r$}.
  \end{matrix}
\end{cases}
$$
\end{result}

Throughout the rest of this section, we always assume that $\de=q^\lam$ and $\frac{m}{2}\le\lam\le m-1$. Recall that the narrow-sense primitive BCH code $\codenp$ has generator polynomial $g_{(q,n,\delta,1)}(x)$. Set $r=m-\lam$. Note that $\de-1$ corresponds to following sequence
$$
\ol{\de-1}=(\underbrace{0,\ldots,0}_{r},\underset{\lam-1}{q-1},q-1,\ldots,q-1).
$$
The key observation in \cite{Mann} is that for $1 \le i \le n-1$, $\al^i$ is a root of $g_{(q,n,\delta,1)}(x)$ if and only if the sequence of $i$ has a circular run of length at least $r$ with respect to $0$. Similarly, note that $n-\de+1$ corresponds to the following sequence
$$
\ol{n-\de+1}=(\underbrace{q-1,\ldots,q-1}_{r},\underset{\lam-1}{0},0,\ldots,0).
$$
Therefore, for $1 \le i \le n-1$, $\al^i$ is a root of $g_{(q,n,\delta,n-\delta+1)}(x)$ if and only if the sequence of $i$ has a circular run of length at least $r$ with respect to $q-1$. The following proposition presents the degree of $g_{(q,n,\delta,1)}(x)$ and $g_{(q,n,\delta,n-\de+1)}(x)$.

\begin{result}{\rm \cite[p. 155]{Mann}}\label{res-deg}
Set $r=m-\lam$. Then
$$
\deg(g_{(q,n,\delta,1)}(x))=\deg(g_{(q,n,\delta,n-\delta+1)}(x))=l_r(m)-1+(q-1)^2\sum_{u=0}^{r-2}(r-u-1)(q^{m-r-u-2}-l_r(m-r-u-2)).
$$
\end{result}

We have the following estimation on the dimension of $\ocode$.

\begin{theorem}
Set $r=m-\lam$. Then $\ocode$ has parameters $[q^m-1,k,d \ge 2\de]$, in which the dimension $k$ satisfies
$$
k \ge q^m-2l_r(m)+2l_r(m-r)-2(q-1)^2\sum_{u=0}^{r-2}(r-u-1)(q^{m-r-u-2}-l_r(m-r-u-2)),
$$
and
$$
k \le q^m-2l_r(m)+ml_r(m-r)-2(q-1)^2\sum_{u=0}^{r-2}(r-u-1)(q^{m-r-u-2}-l_r(m-r-u-2)).
$$
\end{theorem}
\begin{proof}
Since $\lam\ge\frac{m}{2}$, we have $m\ge 2r$. Define a set $N=\{1 \le i \le n-1 \mid g_{(q,n,\delta,1)}(\al^i)=g_{(q,n,\delta,n-\de+1)}(\al^i)=0\}$. Since $\gx=(x-1)\lcm(g_{(q,n,\delta,1)}(x),g_{(q,n,\delta,n-\de+1)}(x))$, we have
$$
\deg(\gx)=\deg(g_{(q,n,\delta,1)}(x))+\deg(g_{(q,n,\delta,n-\de+1)}(x))+1-|N|.
$$
Since $\deg(g_{(q,n,\delta,1)}(x))$ and $\deg(g_{(q,n,\delta,n-\de+1)}(x))$ are known by Result~\ref{res-deg}, it suffices to estimate the size of $N$. $N$ contains the number $1 \le i \le n-1$, such that $\ol{i}$ contains two runs of length $r$ with respect to $0$ and $q-1$, where at most one of them is a circular run. Let $N^{\pr}$ be the set of integers $1 \le i \le n-2$ such that the first $r$ entries of $\ol{i}$ is a straight run of length $r$ with respect to $0$ and the last $m-r$ entries contain a straight run of length $r$ with respect to $q-1$. Clearly, we have $|N^{\pr}|=l_r(m-r)$. Note that each element of $N$ is a proper cyclic shift of an element of $N^{\pr}$. Moreover, for each $i\in N^{\pr}$, we have
\begin{equation*}
2\le|\{q^ji \bmod{n} \mid 0 \le j \le m-1\}|\le m,
\end{equation*}
which implies
$$
2|N^{\pr}| \le |N| \le m|N^{\pr}|.
$$
Thus, the conclusion follows from a direct computation.
\end{proof}

\subsection{The minimum distance of LCD BCH codes $\ocode$}

While it is difficult to determine the dimension of LCD BCH codes in general, it is more difficult to find out the minimum distance of LCD BCH codes. For the code  $\ocode$, the BCH bound $d \geq 2 \delta$ is usually very tight. But it would be better if we could determine the minimum distance exactly. In this section, we determine the minimum distance $d$ of the code $\ocode$ in some special cases.

Given a codeword $c=(c_0,c_1,\ldots,c_{n-1}) \in \codenp$, we say $c$ is reversible if $(c_{n-1},c_{n-2},\ldots,c_0) \in \codenp$. Namely, $c \in \codenp$ is reversible if and only if $c \in \tcode$. The following theorem says that the reversible codeword in $\codenp$ provides some information on the minimum distance on $\ocode$.

\begin{theorem}\label{thm-mindis}
Let $c(x) \in \codenp$ be a reversible codeword of weight $w$. If $c(1) \ne 0$, then $\ocode$ contains a codeword $(x-1)c(x)$ whose weight is at most $2w$. Therefore the minimum distance $d$ of $\ocode$ satisfies $d\le 2w$. In particular, if the weight of $c(x)$ is $\de$, then the minimum distance $d=2\de$.
\end{theorem}
\begin{proof}
Since $c(x)$ is reversible and $c(1) \ne 0$, we have $(x-1)c(x) \in \ocode$. The weight of $(x-1)c(x)$ is at most $2w$, which implies $d \le 2w$. In particular, if $w=\de$, together with the BCH bound, we have $d=2\de$.
\end{proof}

Let $c(x)=\sum_{i=0}^{n-1}c_ix^i$ be a codeword of a cyclic code $\C$ with length $n$. We can use the elements of $\Fqm^*$ to index the coefficients of $c(x)$. Similarly, let $\ol{\C}$ be the extended cyclic code of $\C$ and let $c(x)=\sum_{i=0}^{n}c_ix^i$ be a codeword of $\ol{\C}$ with length $q^m$. We can use the elements of $\Fqm$ to index the coefficients of $\ol{c(x)}$. The support of $c(x)$ (resp. $\ol{c(x)}$) is defined to be the set of elements in $\Fqm^*$ (resp. $\Fqm$), which correspond to the nonzero coefficients of $c(x)$ (resp. $\ol{c(x)}$).

Given a prime power $q$ and an integer $0 \le s \le q^m-1$, $s$ has a unique $q$-ary expansion $s=\sum_{i=0}^{m-1} s_iq^i$. The $q$-weight of $s$ is defined to be $wt_q(s)=\sum_{i=0}^{m-1}s_i$. Suppose $H$ is a subset of $\Fq^*$, then we use $H^{(-1)}$ to denote the subset $\{h^{-1} \mid h \in H\}$.

The following are two classes of LCD BCH codes whose minimum distances are known.

\begin{corollary}\label{cor-distance}
For the LCD BCH code $\ocode$, we have $d=2 \delta$ if $\de \mid n$.
\end{corollary}
\begin{proof}
It suffices to find a codeword $c(x)$ satisfying the condition in Theorem~\ref{thm-mindis}. If $\de \mid n$, by the proof of \cite[Theorem]{TH}, $\codenp$ contains a reversible codeword $c(x)$ with weight $\de$, where $c(x)=\sum_{i=0}^{\de-1}a_ix^{\frac{ni}{\de}}$ and $c(1)\ne 0$. The desired conclusion then follows from Theorem \ref{thm-mindis}.
\end{proof}

\begin{corollary}
Let $\de=2^r-1$ and $V$ be an $m$-dimensional vector space over $\Ft$. Suppose $1 \le r \le \lf \frac{m}{2} \rf$, then we can choose four $r$-dimensional subspaces of $V$, say $H_i$, $1 \le i \le 4$ of $V$, such that
$H_1 \cap H_2=\{0\}$ and $H_3 \cap H_4=\{0\}$. If $((H_1 \cup H_2) \setminus \{0\})^{(-1)}=(H_3 \cup H_4) \setminus \{0\}$, then the LCD BCH code $\C_{(2,n,\delta,n-\de+1)}$ has parameters $[2^m-1,2^m-2-2m(2^{r-1}-1),2\de]$.
\end{corollary}
\begin{proof}
The dimension of $\C_{(2,n,\delta,n-\de+1)}$ easily follows from Theorem~\ref{thm-gene}. We are going to show that the minimum distance $d=2\de$. Define $\C_{(2,n,\delta,1)}$ (resp. $\C_{(2,n,\delta,n-\de+1)}$) to be the BCH code with length $n=2^m-1$ and generator polynomial $g_{(2,n,\delta,1)}(x)$ (resp. $g_{(2,n,\delta,n-\de+1)}(x)$). Let $\al$ be a primitive element of $\Ftm$. We can assume the zeros of $\C_{(2,n,\delta,1)}$ (resp. $\C_{(2,n,\delta,n-\de+1)}$) include the elements $\{\al^i \mid 1 \le i \le \de-1\}$ (resp. $\{\al^{-i} \mid 1 \le i \le \de-1\}$).

The BCH code $\C_{(2,n,\delta,1)}$ (resp. $\C_{(2,n,\delta,n-\de+1)}$) contains the punctured Reed-Muller code $\RM^+(m-r,m)^*$ (resp. $\RM^-(m-r,m)^*$) as a subcode, in which $\RM^+(m-r,m)^*$ has zeros
$$
\{\al^i \mid 0 < i < 2^m-1, wt_2(i)<r\}
$$
and $\RM^-(m-r,m)^*$ has zeros
$$
\{\al^{-i} \mid 0 < i < 2^m-1, wt_2(i)<r\}.
$$
Let $c=(c_0,c_1,\ldots,c_{n-1})$ be a codeword of $\RM^+(m-r,m)^*$. Since $\RM^+(m-r,m)^*$ is a cyclic code, its coordinates can be indexed in the following way
\begin{equation}\label{eqn-1}
c=(\underset{1}{c_0},\underset{\al}{c_1},\ldots,\underset{\al^{n-1}}{c_{n-1}}),
\end{equation}
where $\sum_{j=0}^{n-1} c_j\al^{ij}=0$ for each $1 \le i \le \de-1$. Similarly, suppose $c^{\pr}=(c_0^{\pr},c_1^{\pr},\ldots,c_{n-1}^{\pr})$ is a codeword of $\RM^-(m-r,m)^*$. Then, its coordinates can be indexed in the following way
\begin{equation}\label{eqn-2}
c^{\pr}=(\underset{1}{c_0^{\pr}},\underset{\al^{-1}}{c_1^{\pr}},\ldots,\underset{\al^{-(n-1)}}{c_{n-1}^{\pr}}),
\end{equation}
where $\sum_{j=0}^{n-1} c_j^{\pr}\al^{-ij}=0$ for each $1 \le i \le \de-1$.

By \cite[Corollary 5.3.3]{AK}, $\RM^+(m-r,m)^*$ contains two minimum weight codewords $c_1(x)$ and $c_2(x)$, such that the support of $c_1(x)$ and $c_2(x)$ are $H_1 \setminus \{0\}$ and $H_2 \setminus  \{0\}$ respectively. Similarly, $\RM^-(m-r,m)^*$ contains two minimum weight codewords $c_3(x)$ and $c_4(x)$, such that the support of $c_3(x)$ and $c_4(x)$ are $H_3 \setminus  \{0\}$ and $H_4 \setminus  \{0\}$ respectively. Moreover, the coordinates of $c_1(x)$ and $c_2(x)$ are arranged in the way of~(\ref{eqn-1}) and the coordinates of $c_3(x)$ and $c_4(x)$ are arranged in the way of~(\ref{eqn-2}). Therefore,
$$
c_1(x)+c_2(x) \in \RM^+(m-r,m)^* \subset \C_{(2,n,\delta,1)}
$$
and
$$
c_3(x)+c_4(x) \in \RM^-(m-r,m)^* \subset \C_{(2,n,\delta,n-\de+1)}.
$$
Since $((H_1 \cup H_2) \setminus \{0\})^{(-1)}=(H_3 \cup H_4) \setminus \{0\}$, by the arrangement of the coordinates of $c_i(x)$, $1 \le i \le 4$, the two codewords $c_1(x)+c_2(x)$ and $c_3(x)+c_4(x)$ coincide. Thus, we have $c_1(x)+c_2(x) \in \tilde{\C}_{(2,n,\delta,n-\de+1)}$. Since $c_1(1)+c_2(1)=0$, we have a codeword $c_1(x)+c_2(x) \in \C_{(2,n,\delta,n-\de+1)}$ with weight $2\de$.
\end{proof}

\begin{example}
Let $q=2$, $m=5$ and $\de=3$ in the above corollary. We are going to show that $\C_{(2,31,6,29)}$ has parameters $[31,20,6]$. Note that the dimension of $\C_{(2,31,6,29)}$ easily follows from Theorem~\ref{thm-gene}, it suffices to prove that the minimum distance is equal to $6$. Let $\al$ be a primitive element of $\gf(2^5)$ and the minimal polynomial of $\al$ over $\Ft$ is $x^5+x^2+1$. Then we have the following four $2$-dimensional subspaces of $\gf(2^5)$:
\begin{align*}
H_1&=\{0,\al,\al^2,\al^{19}\}, \\
H_2&=\{0,\al^8,\al^{12},\al^{18}\}, \\
H_3&=\{0,\al^{12},\al^{13},\al^{30}\}, \\
H_4&=\{0,\al^{19},\al^{23},\al^{29}\}.
\end{align*}
Thus, we have $c_1(x)$ and $c_2(x)$ as codewords of $\C_{(2,31,3,1)}$, whose supports are $H_1 \setminus \{0\}$ and $H_2 \setminus \{0\}$. We have $c_3(x)$ and $c_4(x)$ as codewords of $\C_{(2,31,3,29)}$, whose supports are $H_3 \setminus \{0\}$ and $H_4 \setminus \{0\}$. Clearly, $((H_1 \cup H_2) \setminus \{0\})^{(-1)}=(H_3 \cup H_4) \setminus \{0\}$. Therefore, $c_1(x)+c_2(x)$ coincides with $c_3(x)+c_4(x)$, whose weight is six. Consequently, $c_1(x)+c_2(x) \in \C_{(2,31,6,29)}$ and the minimum distance of $\C_{(2,31,6,29)}$ equals $6$.
\end{example}

Based on our numerical experiment, we have the following conjecture, which can be regarded as an analogy of \cite[Chapter 9, Theorem 5]{MS}.

\begin{conj}
Let $\delta=q^\lambda-1$, where $1 \leq \lambda \leq \lfloor m/2 \rfloor$. Then the code  $\ocode$ has minimum
distance $d=2\delta$.
\end{conj}

\subsection{Parameters of $\ocode$ for small $\delta$}

In this section, we determine the parameters of the code $\ocode$ for a few small values of $\delta$. With the help of Theorem \ref{thm-gene} and Corollary \ref{cor-distance}, we can achieve this in some cases.

Recall that the Melas code over $\gf(q)$ is a cyclic code with length $n$ and generator polynomial $m_{-1}(x)m_1(x)$ and was first studied by Melas for the case $q=2$ \cite{Melas}. The weight distribution of the Melas code has been obtained for $q=2,3$ \cite{GSV,SV}. For $\delta = 2$, the code $\ocode$ is the even-like subcode of the Melas code. The following theorem is a direct consequence of Theorem \ref{thm-gene} and Corollary \ref{cor-distance}.

\begin{theorem}
Suppose $q$ is odd and $ m \ge 2$, then $\C_{(q,n,4,n-1)}$ has parameters $[q^m-1, q^m-2-2m, 4]$.
\end{theorem}

When $\delta = 3$, we have the following result.

\begin{theorem}\label{thm-dl2}
\begin{itemize}
\item[1)] When $q=2$ and $m \geq 4$, $\C_{(q,n,6,n-2)}$ has parameters $[2^m-1, 2^m-2-2m, 6]$.
\item[2)] When $q^m \equiv 1 \pmod{3}$ and $m \geq 4$,  $\C_{(q,n,6,n-2)}$ has parameters $[q^m-1, q^m-2-4m, 6]$.
\end{itemize}
\end{theorem}
\begin{proof}
1) The dimension follows from Theorem \ref{thm-gene}. Applying the BCH and the sphere packing bound, we can see that the minimum distance is $6$.

2) The dimension follows from Theorem \ref{thm-gene}. Since $q^m\equiv 1 \pmod{3}$, we have $3 \mid n$. Therefore, by Corollary \ref{cor-distance}, the minimum distance is $6$.
\end{proof}

\begin{theorem}\label{thm-q3}
Suppose $m \ge 3$, then $\C_{(3,n,8,n-3)}$ has parameters $[3^m-1, 3^m-2-4m, d]$, where $d=8$ if
$m$ is even and $d \geq 8$ if $m$ is odd.
\end{theorem}

\begin{proof}
It follows from Theorem \ref{thm-gene} that the dimension of this code is equal to $q^m-2-4m$. By the BCH
bound, the minimum distance of $\C_{(3,n,8,n-3)}$ is at least $8$.

When $m$ is even, $4$ divides $n$. Hence, the minimum distance of $\C_{(3,n,8,n-3)}$ is equal to $8$ according to Corollary \ref{cor-distance}.
\end{proof}

We have the following conjecture concerning the case $q=3$.

\begin{conj}
When $q=3$, $m \geq 3$ is odd and $\delta=4$,  $\ocode$ has minimum distance $d=8$.
\end{conj}

\begin{example}
Let $q=3$, $m=3$ and $\delta=4$ in Theorem~\ref{thm-q3}. Then $\C_{(3,26,8,n-3)}$ has parameters $[26,13,8]$. According to {\rm \cite[p. 300, Table A.92]{Dingbk15}}, all known ternary linear codes with length $26$ and dimension $13$ has minimum distance at most $8$. Hence, $\C_{(3,n,8,n-3)}$ has the same parameters as the best known linear code.
\end{example}

\section{Concluding remarks}

The main contributions of this paper are the following:
\begin{enumerate}
  \item The characterization of the coset leaders of $q$-cyclotomic cosets $C_j$ modulo $n=q^m-1$, where $1 \le j \le (q-1)q^{\bar m}$. The size of these cyclotomic cosets is also computed.
  \item The determination of the dimension of the LCD BCH codes $\codeeven$ and $\codeodd$ with $\delta=uq^{\bar m}+1$ if $q$ is odd and with $\delta=uq^{\bar m}/2+1$ if $q$ is even, where $1 \le u \le q-1$.
  \item The determination of the dimension of the LCD BCH codes $\codeeven$ and $\codeodd$ when it has designed distance $q^t-1$, where $1 \le t \le \bar m$.
  \item The determination of the dimension of the LCD BCH codes $\mathcal C_{(q, n, 2\delta, n-\delta+1)}$, with $2 \le \de \le q^{(m+1)/2}$ when $m$ is odd and with $2 \le \de \le 2q^{m/2}$ when $m$ is even.
  \item The determination of the dimension of the LCD BCH codes $\mathcal C_{(q, n, 2\delta, n-\delta+1)}$, with $q$ being odd, $\delta=uq^{\bar m}+1$ and $1 \le u \le q-1$.
  \item Lower and upper bounds on the dimension of $\mathcal C_{(q, n, 2\delta, n-\delta+1)}$, where $\delta=q^\lambda$ and $\frac{m}{2} \le \lambda \le m-1$.
\end{enumerate}

Lower bounds on the minimum distance of above codes are derived from the BCH bound. In some special cases, the minimum distances are also determined.

For the two families of LCD BCH codes considered in this paper, we are able to determine their dimensions when $\de$ is relatively small, which is approximately the square root of the length of the code. When $\de$ goes larger, it is much more complicated to compute the size of cyclotomic cosets and to characterize the coset leaders. Hence, there seems no obvious way to extend our results to a larger $\de$.

\section*{Acknowledgements}

The authors are very grateful to the reviewers and the Associate Editor, Prof. Chaoping Xing, for their detailed comments and suggestions that much improved the presentation and quality of this paper.

\end{document}